\documentclass[12pt]{article}
\usepackage[T1]{fontenc}
\usepackage[
left=0.6in,
right=0.6in,
top=0.6in,
bottom=0.6in
]{geometry}   
\usepackage{setspace}
\usepackage{amsmath,amssymb} 
\usepackage{amsthm}
\usepackage{algorithm} 
\usepackage{algpseudocode}
\usepackage{float}
\usepackage{placeins}

\usepackage{amsmath}
\usepackage{amsfonts}
\usepackage{amssymb}
\usepackage{amsthm}
\usepackage{nicefrac}
\usepackage{mathrsfs}

\usepackage{graphicx}
\graphicspath{{media/}}

\usepackage{booktabs}
\usepackage{array}
\usepackage{tabularx}
\usepackage{makecell}
\usepackage{adjustbox}
\usepackage{subcaption}

\usepackage{algorithm}
\usepackage{algpseudocode}

\usepackage{xcolor}
\usepackage{soul}
\usepackage{enumerate}
\usepackage{changepage}
\usepackage{microtype}

\usepackage{url}
\usepackage{hyperref}

\usepackage{parskip}
\usepackage{mathrsfs}
\usepackage{graphicx}
\usepackage{xcolor}
\usepackage{subcaption}
\usepackage{forest}
\usepackage{amsmath}
\usepackage{amssymb}
\usepackage{hyperref}
\usepackage{booktabs}
\usepackage{array}
\usepackage{tabularx}
\usepackage{makecell}
\usepackage{adjustbox}

\newcolumntype{Y}{>{\centering\arraybackslash}X}
\newcolumntype{C}{>{\centering\arraybackslash}X}
\usepackage{algorithm}
\usepackage{algpseudocode}

\algnewcommand{\algorithmicswitch}{\textbf{switch}}
\algnewcommand{\algorithmiccase}{\textbf{case}}

\algdef{SE}[SWITCH]{Switch}{EndSwitch}[1]
  {\algorithmicswitch\ #1\ \algorithmicdo}
  {\algorithmicend\ \algorithmicswitch}

\algdef{SE}[CASE]{Case}{EndCase}[1]
  {\algorithmiccase\ #1:}
  {\algorithmicend\ \algorithmiccase}

\usepackage{amssymb}
\usepackage{amsmath}
\usepackage{array}
\usepackage{changepage}
\usepackage{array}
\usepackage{booktabs}
\usepackage{tabularx}

\newtheoremstyle{myplain}
  {8pt}        
  {8pt}        
  {\itshape}   
  {}           
  {\bfseries}  
  {.}          
  {0.5em}      
  {}           

\theoremstyle{myplain}

\newtheorem{result}{Result}
\newtheorem{observation}{Observation}
\newtheorem{lemma}{Lemma}
\newtheorem{theorem}{Theorem}
\newtheorem{corollary}{Corollary}

\theoremstyle{definition}
\newtheorem{definition}{Definition}

\title{\textsc{Min-Max Uniform Circle Formation by Asynchronous Mobile Robots}}

\author{
Animesh Maiti\thanks{Corresponding author: p22ma201@iitj.ac.in}, Prakhar Shukla, Subhash Bhagat
\\[1ex]
Department of Mathematics\\
Indian Institute of Technology Jodhpur\\
Jodhpur, Rajasthan, India\\[1ex]
}

 \date{}

\begin{document}

\maketitle

\begin{abstract}
 Given a set of point robots $\mathcal{R}$ in the Euclidean plane and a target circle $\mathbf C$ enclosing all robot positions, the \textsc{Min-Max Uniform Circle Formation (MMUCF)} problem requires the robots to move to distinct positions on $\mathbf C$ such that the final configuration forms a regular $n$-gon while minimizing the maximum distance traveled by any robot. Uniform circle formation is a fundamental coordination task in swarm robotics with applications in perimeter monitoring, surveillance, boundary coverage, and pattern formation. The literature does not address the optimization of the maximum individual displacement during the formation process. 
    
In this work, we study the min--max versions of the circle formation and uniform circle formation problems, where the goal is to minimize the maximum distance traveled by any robot. We consider these problems under the $\mathcal{ASYNC}$ model, where robots are autonomous, anonymous, identical, homogeneous, oblivious, and silent, and operate under the \textit{Look--Compute--Move} model with non-rigid motion. We first give necessary conditions for a deterministic solution and then present deterministic, distributed, and collision-free algorithms that form a circle and a uniform circle in finite time while minimizing the maximum movement. The algorithms ensure that robots reach distinct positions on the circle and, in the uniform case, equally spaced positions on $\mathbf C$ under the considered model.   
\end{abstract}
\noindent
\textbf{Keywords:} Uniform Circle Formation, Autonomous Mobile Robots,
Asynchronous Model,
Non-Rigid Movements,
Persistent Lights,
Min--Max Travel-Distance Optimization,
Optimal Assignment.

\section{Introduction}  
In the past two decades, multi-robot systems have become an important research focus because of their robustness, scalability, and versatility in diverse applications~\cite{albiero2021swarm, chun2016robot, kegeleirs2019random, parunak1996applications, steinberg2006intelligent}. A swarm of mobile robots is a distributed multi-agent system consisting of numerous small, \textit{autonomous}, \textit{homogeneous}, and \textit{anonymous} robots that cooperate to accomplish a common task. Each robot operates as an independent computational unit with limited sensing, processing, and actuation capabilities, and coordinated interactions among robots lead to complex collective behavior. Inspired by the self-organizing mechanisms observed in social insects, swarm-based systems are particularly suited for hazardous or inaccessible environments where human intervention is impractical. Their fully decentralized nature enables deployment in unknown and dynamic scenarios and provides robustness against individual failures. In addition, the relatively simple design of individual robots makes such systems cost-effective and fault-tolerant compared to relying on a single sophisticated robot. Consequently, multi-robot systems have been utilized in disaster response and search-and-rescue missions~\cite{brambilla2013swarm}, autonomous exploration and mapping~\cite{corke2005networked}, surveillance and environmental monitoring~\cite{schwager2009decentralized}, agriculture and mining in hazardous environments~\cite{bayindir2016review}, education and construction~\cite{brambilla2013swarm}, and coordinated geometric formation tasks requiring robots to surround a region or target~\cite{schwager2009decentralized, brambilla2013swarm}.

In distributed computing, theoretical research on multi-robot systems has focused on coordination and complexity issues arising in systems of autonomous mobile entities operating without centralized control. Several computational models emphasizing different capabilities have been proposed, including models with explicit communication~\cite{santorobook2}, memory~\cite{flocchini2013rendezvous}, or identifiers~\cite{terai2018gathering}.

\noindent\textbf{Our Model.} In this paper, we consider the \textit{oblivious} (no memory of past actions) mobile robots model. In this model, robots are simple computational entities represented as points in the Euclidean plane and operate in repeated \textit{Look-Compute-Move} cycles. During each cycle, an active robot observes the positions of other robots in its local coordinate system (Look), computes a destination based on the observed configuration (Compute), and moves toward that destination (Move). Movements may be \textit{rigid}, meaning the robot always reaches its destination, or \textit{non-rigid}, meaning the robot may be stopped before reaching the destination but is guaranteed to move at least a minimum distance $\delta$ toward it.

The robots are \textit{autonomous} (operate without centralized coordination), \textit{homogeneous} (identical in capabilities and behavior), \textit{anonymous} (without unique identifiers), \textit{oblivious} (no memory of past actions), and \textit{silent} (no explicit communication). Robots do not retain any memory between computational cycles. In extended models, they may be equipped with a constant amount of persistent memory~\cite{flocchini2013rendezvous}, usually implemented through externally visible lights, commonly known as the \textit{luminous model}, with a finite set of colors~\cite{sharma2017constant}. Such persistent memory can be used for communication and state retention (FALL model)~\cite{di2014robots}, only for remembering the robot’s internal state (FSTATE model)~\cite{di2017mutual}, or only for communication without storing past states (FCOMM model)~\cite{flocchini2013rendezvous}. Each robot uses its own local Cartesian coordinate system centreed at its current position, with its own unit of length and orientation of the coordinate axes. Depending on the model~\cite{Suzuki1999}, different levels of agreement among robots may be assumed. In the strongest case, robots agree on the directions and orientations of all axes, forming a common global reference frame. In weaker settings, robots may agree only on the direction and orientation of one axis. Robots may also share chirality, meaning they have a common notion of clockwise direction even if their axes are not aligned. In the weakest case, robots may have no agreement on axis directions or orientations, and their coordinate systems are completely independent. In this paper, we consider robots with one axis of agreement and equipped with visible lights.

Robot activations are controlled by an adversarial but fair scheduler under one of the standard models. In the fully synchronous model ($\mathcal{FSYNC}$), all robots execute their cycles simultaneously in discrete rounds. In the semi-synchronous model ($\mathcal{SSYNC}$), an arbitrary subset of robots is activated in each round. In the asynchronous model ($\mathcal{ASYNC}$), there is no common notion of time, and the duration of each phase is finite but unpredictable. Fairness guarantees that each robot is activated infinitely often, allowing progress to be measured in rounds or epochs depending on the synchronization model.

\noindent\textbf{Problem Setting.}
Let $\mathbf{C}$ be a circle of finite radius $h$, and let $\mathcal R=\{r_1,\ldots,r_n\}$ be a set of $n$ robots initially placed at distinct positions inside $\mathbf{C}$. Let $\mathcal P=\{p_1,\ldots,p_n\}$ be a set of $n$ target points on the boundary of $\mathbf{C}$. The goal is to move the robots so that each robot occupies a distinct point of $\mathcal P$ while minimizing the maximum distance traveled by any robot. In the \textsc{MMUCF} problem, $\mathcal P$ consists of the vertices of a regular $n$-gon inscribed in $\mathbf{C}$, i.e., the points are equally spaced along the circumference. Bhattacharya et al.~\cite{bhattacharya2009optimal} studied a centralized version of this problem in which robot movements are computed centrally. In contrast, we consider a distributed setting where robots act autonomously and make decisions based only on their local observations.

\noindent\textbf{Our Contributions.} In this problem, we prove three impossibility results in Section~\ref{impossibility} (Theorems~\ref{theorem-3}–\ref{theorem-5}), showing that it cannot be solved without additional assumptions about the robots. In Section~\ref{Minmaxunifrom}, we then design a deterministic algorithm \textsc{MinMaxUniformCircle()} for robots equipped with visible lights and one-axis agreement that achieves uniform circle formation under the min--max objective. These results indicate that the uniform variant is strictly more difficult than the non-uniform variant and requires additional capabilities to be solvable.


\noindent\textbf{Related Works.}
\noindent Geometric formation problems for autonomous mobile robots have been extensively studied in the distributed computing literature since the seminal work of Suzuki and Yamashita~\cite{Suzuki1999}, which introduced the framework for pattern formation by anonymous robots operating without centralized control. Subsequent research by Flocchini et al.~\cite{flocchini1999hard}, Défago and Konagaya~\cite{Defago2002}, Défago and Souissi~\cite{defago2008non}, and Flocchini et al.~\cite{DBLP:series/synthesis/2012Flocchini} established feasibility and impossibility results for circle formation and related tasks under different synchronization models and visibility constraints. In the circle formation problem, robots must move to the boundary of a common circle, and the problem has been studied under asynchronous settings, limited visibility, and weak agreement on coordinate systems. A stricter variant, Uniform Circle Formation ($\mathcal {UCF}$), requires robots to occupy equally spaced positions on the circle and has received significant attention due to its symmetry requirements. Deterministic solutions for $\mathcal {UCF}$ have been obtained even under the fully asynchronous model by Flocchini et al.~\cite{flocchini2017distributed} and Viglietta~\cite{vigl2016}, resolving long-standing open cases. Several extensions also consider additional capabilities such as persistent visible lights, obstructed visibility, fat robots, and rigid movements; for example, $\mathcal {UCF}$ with persistent lights and rigid motion has been solved by Pattanayak et al.~\cite{pattanayak2024time}. Other works study trade-offs between time, memory, and communication complexity, including results by Feletti et al.~\cite{feletti2018uniform, feletti2024optimal} and Mondal et al.~\cite{mondal2019uniform}. However, most existing works focus on feasibility and symmetry breaking rather than optimizing movement costs, and relatively little attention has been given to minimizing the maximum distance traveled by any robot in fully distributed and oblivious settings. Bhagat et al.~\cite{bhagat2018optimum} study this objective in a distributed setting for circle formation, whereas Bhattacharya et al.~\cite{bhattacharya2009optimal} consider a centralized formulation. Table~\ref{tab:earlier-work} summarizes the existing results on circle formation and uniform circle formation under different assumptions.

\begin{table}[htbp]
\caption{Related work on circle formation and uniform circle formation.}
\label{tab:earlier-work}
 \centering

\begingroup
\renewcommand{\arraystretch}{1.15}
\setlength{\tabcolsep}{3pt}
\small

\begin{tabularx}{\textwidth}{
    @{}
    >{ \centering\arraybackslash}p{1.05cm}
    C
    >{ \centering\arraybackslash}p{2.45cm}
    C
    C
    >{ \centering\arraybackslash}p{2.65cm}
    @{}
}
\toprule
\textbf{Problem}
&
\textbf{Input circle}
&
\textbf{Additional capabilities}
&
\textbf{Movement}
&
\textbf{Optimization criterion}
&
\textbf{Reference}
\\
\midrule

$\mathcal{CF}$
&
No
&
Lights
&
Non-rigid
&
Yes
&
Solved in~\cite{bhagat2018optimum}
\\

$\mathcal{UCF}$
&
No
&
None
&
Non-rigid
&
No
&
Solved in~\cite{flocchini2017distributed}
\\

$\mathcal{UCF}$
&
No
&
Persistent lights
&
Rigid
&
No
&
Solved in~\cite{pattanayak2024time}
\\

\textbf{$\mathcal{UCF}$}
&
\textbf{Yes}
&
\textbf{One axis and lights}
&
\textbf{Non-rigid}
&
\textbf{Yes}
&
\textbf{This paper}
\\

\bottomrule
\end{tabularx}
\endgroup
\end{table}

\section{Preliminaries}

Let $r(t)$ denote the location of robot $r_i\in\mathcal R$ at time $t$, and let $\mathcal  R(t)$ denote the configuration of robots at time $t$, defined as the multiset of occupied positions, i.e., $\mathcal  R(t)=\{r_1(t),r_2(t),\cdots,r_n(t)\}$. Let $\widetilde{\mathcal R}$ denote the set of all such robot configurations. Robots operate in repeated {\it Look-Compute-Move} cycles under the $\mathcal{ASYNC}$ model, and all robots execute the same algorithm. For each robot $r_i\in\mathcal R$, let $rad_ R(t)$ denote the half-line originating at $\mathcal O$ and passing through $r(t)$ at time $t$ , where $\mathcal O$ is the centre of $\mathbf C$. Let $\widetilde{\mathbf{C}}$ denote the set of all circles with finite positive radius, and let $\widetilde{\mathcal G}$ denote the set of all algorithms that solve the {\it uniform circle formation} problem over any given circle ${\mathbf{C}}\in\widetilde{\mathbf{C}}$. Note that $\widetilde{\mathcal G}\neq\emptyset$ \cite{feletti2018uniform}. Consider an algorithm $\mathcal A\in\widetilde{\mathcal G}$, a circle $\mathbf C$, and an initial robot configuration $\mathcal  R(t_0)$. Let $p_i$ be the final position of robot $r_i$ on $\mathbf C$ during the execution of $\mathcal A$ starting from $\mathcal  R(t_0)$, so that the points $\{p_1,p_2,\ldots,p_n\}$ form a regular $n$-gon. Then there exists $t^*\ge t_0$ such that $\mathcal R(t^*)=\{p_1,p_2,\ldots,p_n\}$. Let $l_i(\mathcal A, {\mathbf{C}})$ denote the path traversed by robot $r_i$ from $r_i(t_0)$ to $p_i$, and let $|l_i(\mathcal A, {\mathbf{C}})|$ denote its length. Define $\lambda_{\mathbf{C}}(\mathcal A,\mathcal  R(t_0))$ as the maximum distance traveled by any robot during the execution of $\mathcal A$, i.e., $\lambda_{\mathbf{C}}(\mathcal A,\mathcal  R(t_0))=\max\{|l_i(\mathcal A, {\mathbf{C}})|: r_i\in\mathcal R\}.$ A uniform circle formation algorithm $\bar{\mathcal A}\in\widetilde{\mathcal G}$ solves the \textsc{MMUCF} problem if $\lambda_{\mathbf{C}}(\bar{\mathcal A},\mathcal  R(t_0))
=\min\{\lambda_{\mathbf{C}}(\mathcal A, \mathcal  R(t_0)):\mathcal A\in\widetilde{\mathcal G}\}$ for all ${\mathbf{C}}\in\widetilde{\mathbf{C}}$ and $\mathcal  R(t_0)\in\widetilde{\mathcal R}$. When there is no ambiguity, we denote $\lambda_{\mathbf{C}}(\bar{\mathcal A},\mathcal  R(t_0))$ by $\lambda^*_{\mathbf{C}}$. Let $\widetilde{\mathcal G^*}$ denote the set of all \textsc{MMUCF} algorithms over $\widetilde{\mathbf{C}}$.

Bhattacharya et al.~\cite{bhattacharya2009optimal}, and Chen et al.~\cite{chen2015optimal} presented centralized deterministic algorithms that solve the \textsc{MMUCF} problem over $\widetilde{\mathbf{C}}$. These algorithms also solve the problem under the $\mathcal{FSYNC}$ model with unlimited visibility, agreement on both coordinate axes, and rigid movements, implying $\widetilde{\mathcal G^*}\neq\emptyset$. According to these results, there exists at least one robot position $r_i(t_0)$ such that $l_i(\mathcal A, {\mathbf{C}})$ is a straight-line segment $\overline{r_i(t_0)p_i}$ and $\lambda^*_{\mathbf{C}}=|l_i(\mathcal A, {\mathbf{C}})|$. Let $u_ R(t)$ and $v_ R(t)$ be the nearest and farthest points on $\mathbf C$ from $r(t)$, respectively; these are the intersection points of $\mathbf C$ with the line joining $r(t)$ and $\mathcal O$. We next state two known results that characterize the optimal value $\lambda^*_{\mathbf{C}}$ and are used later in our analysis.

\begin{result}\cite{bhattacharya2009optimal} 
\label{result-1}
For a given circle $\mathbf C$ and initial configuration $\mathcal  R(t_0)$, the value of $\lambda^*_{\mathbf{C}}$ is unique.
\end{result}

\begin{result}\cite{chen2015optimal}  
\label{result-2}
Let $\bar{\mathcal A}$ be an algorithm for the \textsc{MMUCF} problem with $\lambda^*_{\mathbf{C}}=\overline{r_i(t_0)p_i}$. Then either $p_i=u_i(t_0)$ or there exists another robot $r_j(t_0)$, $j\neq i$, such that $\lambda^*_{\mathbf{C}}=\overline{r_j(t_0)p_j}$.
\end{result}

\begin{figure}[ht]
    \centering
    \vspace{-0.3cm}
    
    \includegraphics[width=0.4\textwidth]{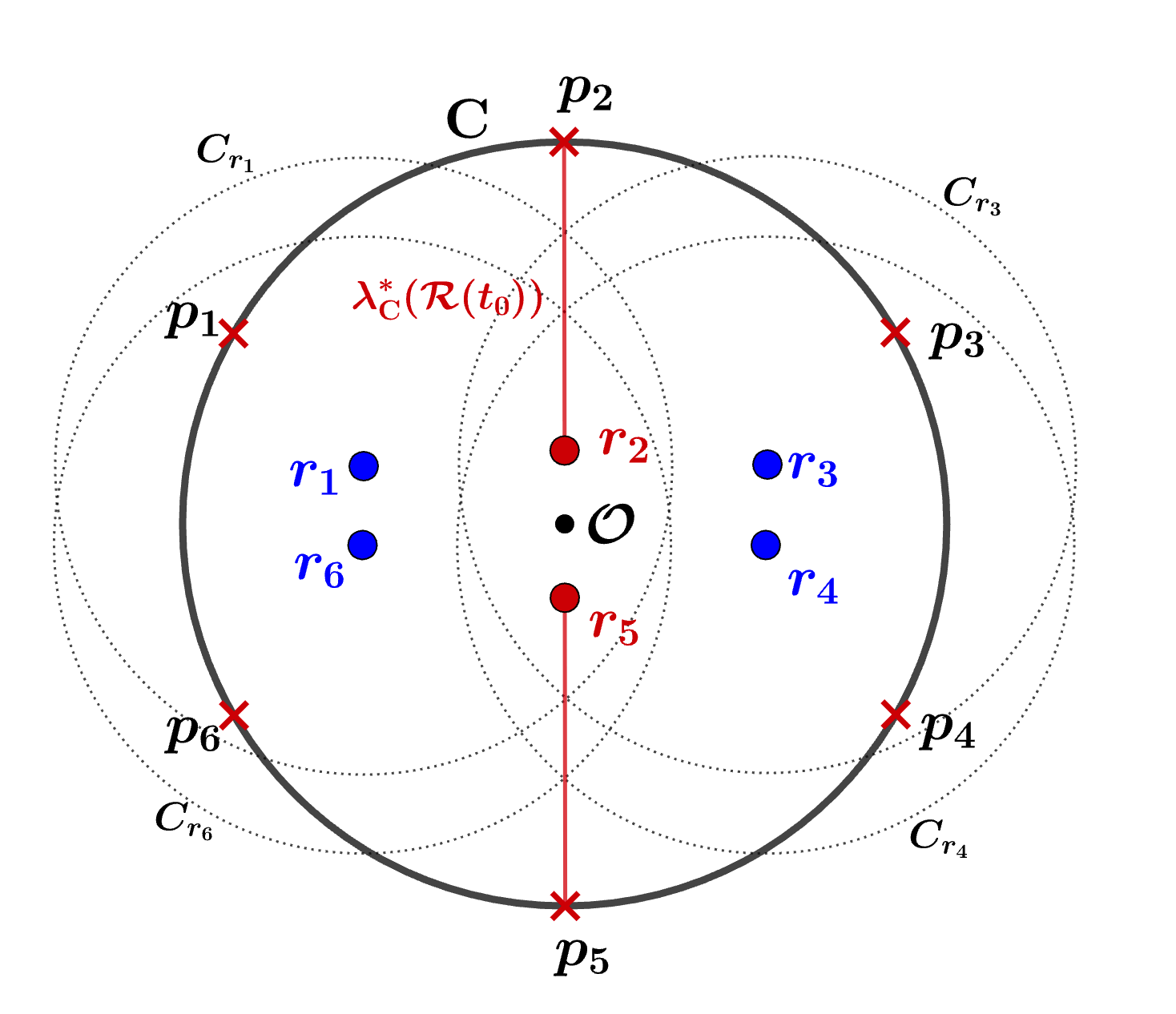}
   
     \caption{ An illustration of target point set and corresponding assignments: initial positions of six robots $r_1,\ldots,r_6$ and the corresponding target set is $\mathcal P=\{p_1,\ldots,p_6\}$. Here, $\lambda^*_{\mathbf{C}}(\mathcal R(t_0))=|\overline{r_2p_2}|=|\overline{r_5p_5}|$ and $\{r_1,r_5\}$ is the {\it critical set}. There are two different assignments for the given configuration: (i) $\{r_1\rightarrow p_1, r_2\rightarrow p_2,r_3\rightarrow p_3,r_4\rightarrow p_4,r_5\rightarrow p_5,r_6\rightarrow p_6\}$ and (ii) $\{r_1\rightarrow p_6, r_2\rightarrow p_2,r_3\rightarrow p_4,r_4\rightarrow p_3,r_5\rightarrow p_5,r_6\rightarrow p_1\}$}
     \label{def-1}
\end{figure}

Let $C_{r}(t)$ denote the circle centreed at $r(t)$ with radius $\lambda^*_{\mathbf{C}}$. Let $p(t)$ and $q(t)$ be the intersection points of $C_{r}(t)$ with $\mathbf C$, and let $\mathcal  R(t)$ be the set of all such intersection points for all robots. A robot position $r(t)$ is called {\it pivotal} if $\lambda^*_{\mathbf{C}}=\overline{r(t)p_i}$. The set of all pivotal robot positions is denoted $\mathcal R'(t)$. A subset $\mathcal R^*(t)\subseteq \mathcal R'(t)$ is called a {\it critical set} if it is minimal such that $\lambda^*_{\mathbf{C}}(\mathcal  R(t))> \lambda^*_{\mathbf{C}}(\mathcal  R(t)\setminus \mathcal R^*(t))$. A set of co-circular points $\mathcal P$ on $\mathbf C$ is a {\it regular point set} if its convex hull forms a regular $k$-gon where $|\mathcal P|=k$. An assignment of robots to $\mathcal P$ is a bijection $f:\mathcal  R(t)\rightarrow \mathcal P$. The set $\mathcal P$ is called the {\it target point set} if there exists an assignment $f$ such that $\lambda^*_{\mathbf{C}}(\mathcal  R(t))=\max\{dist(r(t),f(r(t))):r(t)\in\mathcal  R(t)\}.$ Such an assignment is called an {\it optimal assignment}. The target point set is unique, though the optimal assignment may not be.

\noindent\textbf{Computing the target point set~\cite{bhattacharya2009optimal}.}
\label{app:target-computation}
We compute the {\it target point set} $P^*$, corresponding to  $\lambda^*_C(\mathcal R(t))$ in the following way:
\begin{itemize}
 \item For each robot position $r_i(t)\in\mathcal R(t)$, draw circle $C_i(t)$ having centre at $r_i(t)$ and radius $\lambda^*_C(\mathcal R(t))$. The circle $C_i(t)$ intersect the circle $C$ at least one point. Let $\mathcal C$ be the set of all such circles. Let $p_i(t)$ and $q_i(t)$ be the intersection points between circles $C_i(t)$ and $C$ (these two points may be the same if these two circles intersect exactly at one point). Let $B$ be the set of all such intersection points for the robots positions in $\mathcal R(t)$.
 \item For each $a_x\in B$, do the following:
 \begin{itemize}
  \item compute the point set $W=\{a_1,a_2,\ldots,a_n\}$ such that all $a_i$ lies on $C$, $a_1=a_x$ and these points form a regular $n$-gon.
  \item compute a bipartite graph $G_x=(V_1\cup V_2, E)$ where $V_1=B$, $V_2=\mathcal C$ and there is an edge in $E$ between $a_l\in V_1$ and $C_j\in V_2$ iff $a_l$ lies within or on the circle $C_j$.
  \item If there is a perfect matching in the graph $G_x$, then by lemma 2, we can conclude that $W=P^*$ where the perfect matching provides us an {\it optimum assignment}.
 \end{itemize}

\end{itemize}

\section{Impossibility Results for \textsc{MMUCF} Problem} 
\label{impossibility}
In this section we present two impossibility results for the \textsc{MMUCF} problem. These impossibility results are with respect to the minimal capabilities of the robots required to solve the problem. According to the \textbf{Theorem} \ref{result-2}, during the execution of each algorithm $\mathcal A\in\widetilde{\mathcal G^*}$, first robots require to compute their final destination points on the circle $\mathbf C$ and then accordingly move the towards respective destination points. A robot can use the algorithm presented in \cite{bhattacharya2009optimal} or in \cite{chen2015optimal} to compute its final destination point. However, this simple strategy can not be adopted directly to a swarm of robots. There are many factors that forbid the direct adaptation of this strategy, like the activation schedule of the robots, the non-rigid movements of the robots, the lack of a global coordinate system, the obliviousness of the robots, etc.

\begin{observation}
\label{obs-centre-ambiguity}
Let $\mathcal{R}(t)$ be a configuration for the \textsc{MMUCF} problem on a target circle $\mathbf C$ with centre $\mathcal O$. Suppose there exists a robot $r_c(t)\in\mathcal{R}(t)$ such that $r_c(t)=\mathcal O$, and $r_c$ is a critical robot. Then the corresponding regular target point set on $\mathbf C$ may not be unique, although the optimal min--max value $\lambda_{\mathbf{C}}^*(\mathcal{R}(t))$ is unique, because every point on $\mathbf C$ is at the same distance from $\mathcal O$.
\end{observation}

\begin{theorem}
 \label{theorem-3}
In the absence of any agreement on the local coordinate axes and common {\it chirality}, robots cannot agree on a common {\it optimal assignment}.
\end{theorem}

\begin{figure}[ht!]
     \centering
     \vspace{-0.3 cm}
     
     \includegraphics[width=0.4\textwidth]{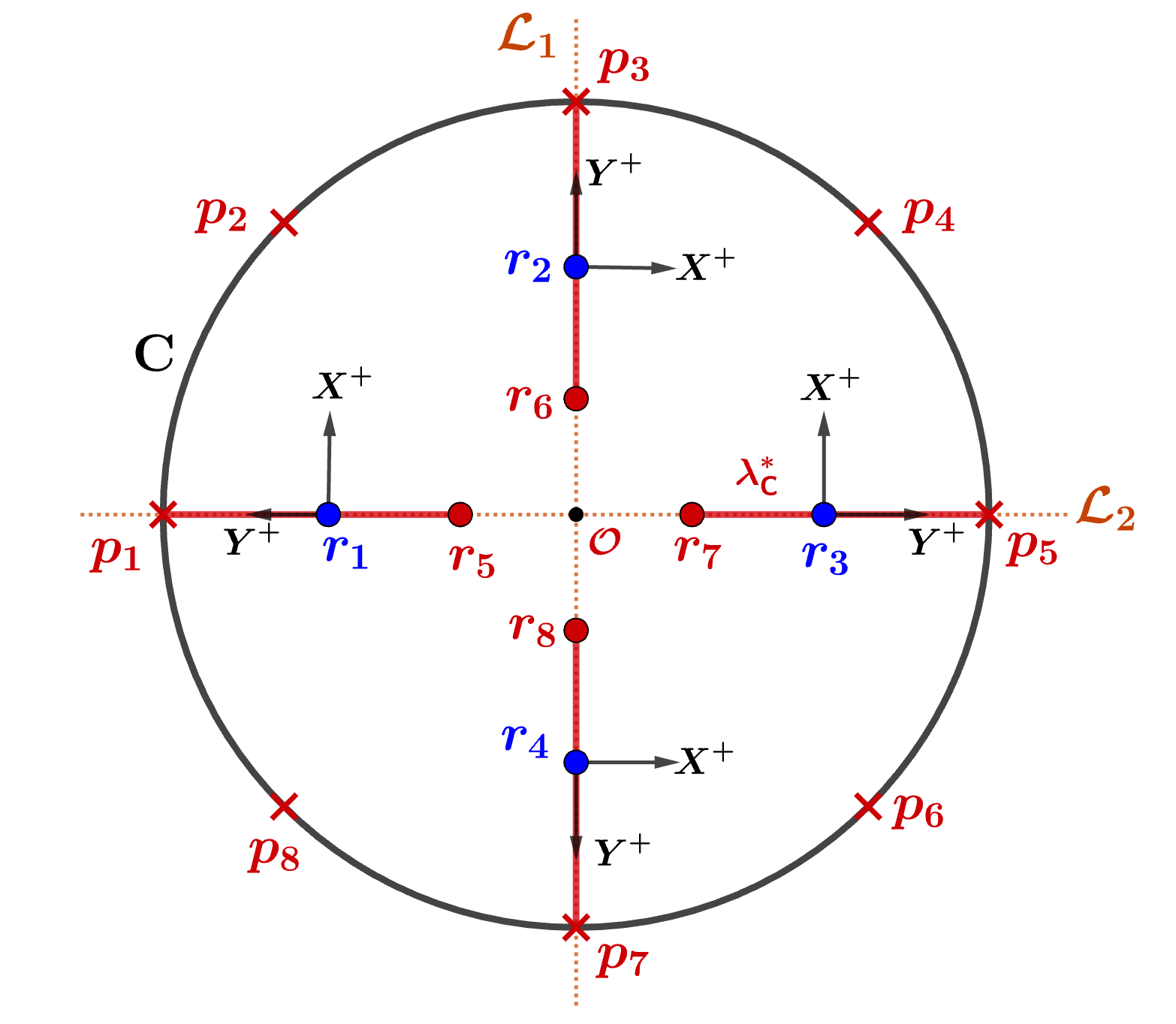}
     
     \caption{Illustration for the proof of \textbf{Theorem}~\ref{theorem-3}. In this symmetric configuration of $\mathcal{R}=\{r_1, r_2, \ldots, r_8\}$, the min-max distance is $\lambda_{\mathbf{C}}^*(t_0) = d(r_5, p_1) = d(r_6, p_3) = d(r_7, p_5) = d(r_8, p_7)$. Consequently, the critical robots are $r_5, r_6, r_7,$ and $r_8$.}
     \label{count-01}
 \end{figure}

\begin{proof}
 Let $\mathcal R=\{r_1,\ldots,r_8\}$ be a set of robots operating under {\it $\mathcal{FSYNC}$} model. Let $\mathbf C$  be a given circle with $\mathcal O$ as its centre. Consider an initial configuration $\mathcal  R(t_0)=\{r_1(t_0),\ldots,r_8(t_0)\}$ of robots in $\mathcal R$ as depicted in Figure \ref{count-01} such that  $\mathcal  R(t_0)$ has two lines of symmetry $\mathcal L_1$ and $\mathcal L_2$. In this case, $\lambda^*_{\mathbf{C}}(\mathcal  R(t_0))=|\overline{r_5p_1}|$. Let $\mathcal A$ be an {\it optimal assignment} algorithm. Algorithm $\mathcal A$ should have following assignments: $r_5\rightarrow p_1$, $r_6\rightarrow p_3$, $r_7\rightarrow p_5$ and $r_8\rightarrow p_7$. Now, we argue about the rest of the assignments. Since algorithm $\mathcal A$ outputs an {\it optimal assignment}, the assignments of the robots $r_1,r_2,r_3$ and $r_4$ should belong to the sets $\{p_2,p_8\}$, $\{p_2,p_4\}$,$\{p_4,p_6\}$, and $\{p_6,p_8\}$ respectively. For the robot configuration $\mathcal  R(t_0)$, the local views of the robots $r_1$ and $r_3$ are symmetric about the line $\mathcal L_1$ and the local views of the robots $r_2$ and $r_4$ are symmetric about the line $\mathcal L_2$. This implies that whatever assignment algorithm $\mathcal A$ makes for the robots $r_1$ and $r_3$, these two points are images of each other about the line $\mathcal L_1$. A similar argument can be made for the assignments of $r_2$ and $r_4$ about the line $\mathcal L_2$. Thus, algorithm $\mathcal A$ assigns  $r_1$ and $r_3$ to either $p_2$ and $p_4$ respectively or $p_8$ and $p_6$ respectively. Without loss of generality, suppose $r_1$ and $r_3$ are assigned to $p_2$ and $p_4$ respectively. Then, algorithm $\mathcal A$ can not assign robot $r_2$ any of the points in $\{p_2,p_4\}$. This contradicts that $\mathcal A$ is an {\it optimal assignment} algorithm. Hence, the robots can not agree on a common {\it optimal assignment}. 
 \end{proof}

\begin{definition}[$k$-fold rotational symmetry]
A configuration $\mathcal R(t)$ has \emph{$k$-fold rotational symmetry}
about the centre $\mathcal O$ of $\mathbf C$ if a rotation
$\rho$ by angle $2\pi/k$, $k\geq2$, maps $\mathcal R(t)$ onto itself.
For a robot position $r_i\in\mathcal R(t)$, the set
$\mathbf{Orbit}_{r_i}
=
\{r_i,\rho(r_i),\rho^2(r_i),\ldots,\rho^{k-1}(r_i)\}$
is called the \emph{rotational orbit} of $r_i$ under $\rho$.
\end{definition}

 \begin{theorem}
  \label{theorem-4}
 In absence of any agreement on the local coordinate systems, the \textsc{MMUCF} problem is unsolvable in $\mathcal{FSYNC}$ model for the robots having non -rigid movements, even if robots have common chirality and lights (\textsc{Fall} light model).
 \end{theorem}

\begin{figure*}[ht!]
    \centering
    \captionsetup[subfigure]{justification=centering}

    \begin{subfigure}[t]{0.44\textwidth}
        \centering
        \includegraphics[width=0.9\linewidth]{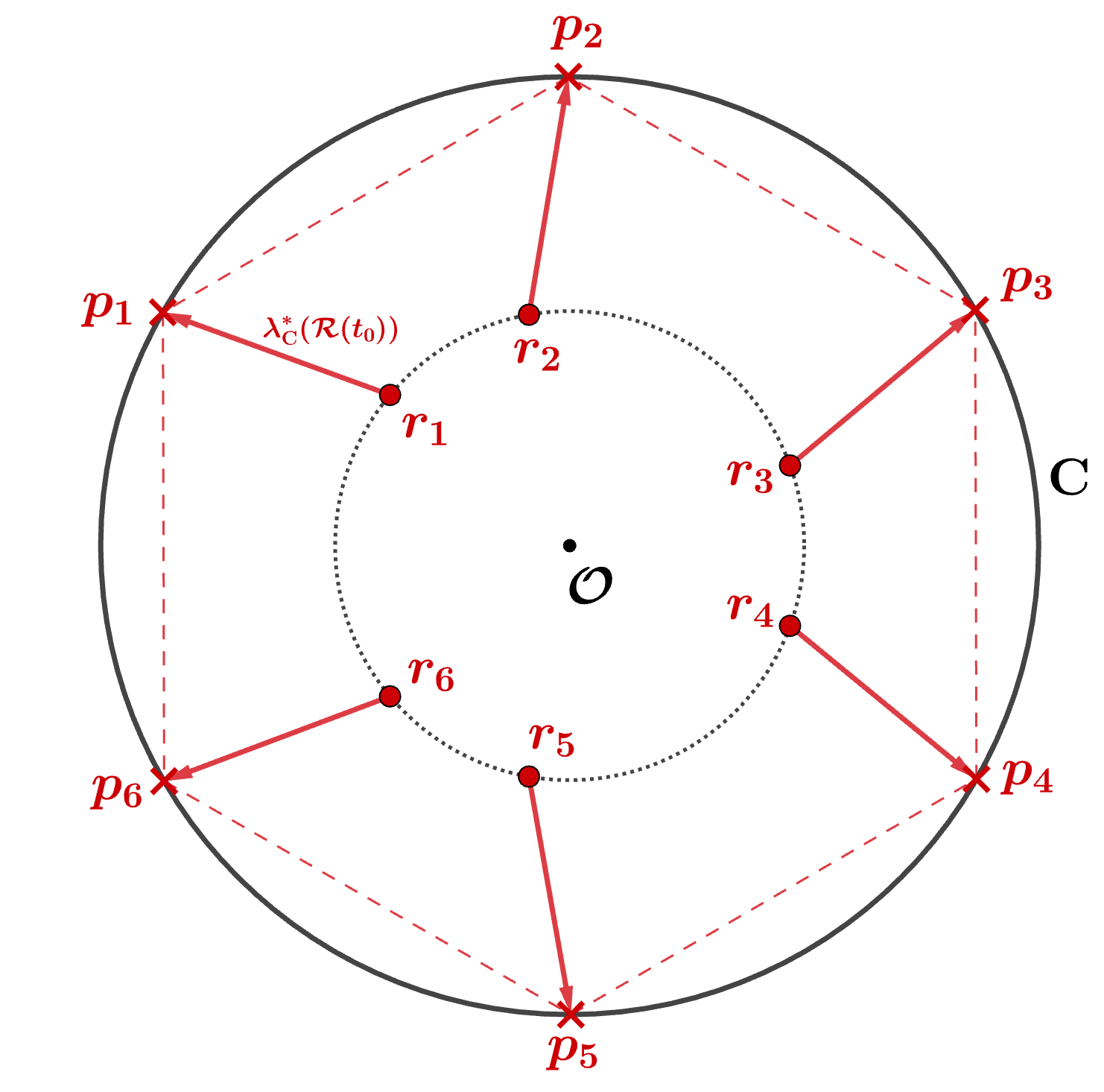}
        \caption{}
        \label{fig:without-agreement-a}
    \end{subfigure}
    \hfill
    \begin{subfigure}[t]{0.44\textwidth}
        \centering
        \includegraphics[width=0.9\linewidth]{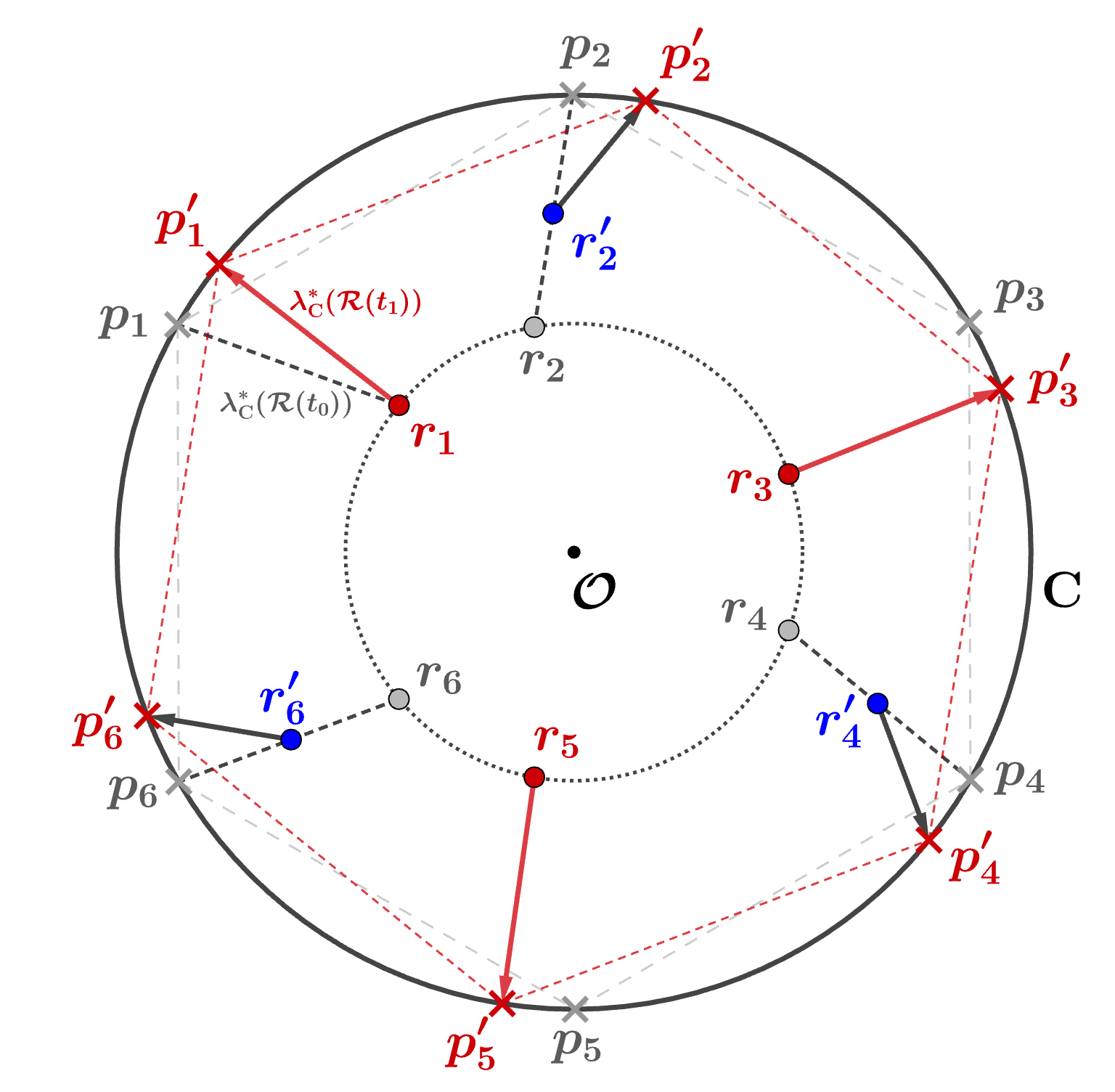}
        \caption{}
        \label{fig:without-agreement-b}
    \end{subfigure}

    \caption{Illustration for the proof of
    Theorem~\ref{theorem-4}.
    (a) A symmetric initial configuration
    $\mathcal R(t_0)=\{r_1(t_0),\ldots,r_6(t_0)\}$,
    where the optimal min-max distance is
    $\lambda_{\mathbf C}^*(t_0)
    =d(r_i(t_0),p_i)$ for $i=1,\ldots,6$.
    (b) A symmetric configuration
    $\mathcal R(t)=\{r_1(t),\ldots,r_6(t)\}$,
    where the optimal min-max distance is
    $\lambda_{\mathbf C}^*(t)
    =d(r_1(t_0),p_1')
    =d(r_3(t_0),p_3')
    =d(r_5(t_0),p_5')$.}
    \label{Imp_2}
\end{figure*}
 
 \begin{proof}
Let $\mathcal R=\{r_1,r_2,\ldots,r_6\}$ be a set of robots operating under  fully-synchronous  \textsc{Fall} light model. Suppose robots in $\mathcal R$ have non-rigid movements and common chirality.  Let $\mathbf C$ be a given circle with $\mathcal O$ as its centre. Consider an initial configuration $\mathcal  R(t_0)=\{r_1(t_0),r_2(t_0),\ldots,r_6(t_0)\}$ of robots in $\mathcal R$ as depicted in Figure \ref{Imp_2}(A) such that (i) $\mathcal  R(t_0)$ is rotationally symmetric around $\mathcal O$ with $120^o$ angle of rotation (ii) all the robots are equidistant from $\mathcal O$. Consider another configuration  $\mathcal  R(t)=\{r_1(t),r_2(t),\ldots,r_6(t)\}$, $t>t_0$, of robots in $\mathcal R$ as depicted in Figure \ref{Imp_2}(B) such that (i) $\mathcal  R(t)$ is rotationally symmetric around $\mathcal O$ with $120^o$ angle of rotation (ii) robots $r_1$, $r_3$ and $r_5$ are equidistant from $\mathcal O$ and robots $r_2$, $r_4$ and $r_6$ are equidistant from $\mathcal O$ (iii) $dist(r_2,\mathcal O)>dist(r_1,\mathcal O)$. If we run algorithm proposed in~\cite{bhattacharya2009optimal}, we get (i) for $\mathcal  R(t_0)$, $\lambda^*_{\mathbf{C}}(\bar{\mathcal A},\mathcal  R(t_0))=d(p_1,r_1(t_0))$ and (ii) for $\mathcal  R(t)$, $\lambda^*_{\mathbf{C}}(\bar{\mathcal A},\mathcal  R(t))=d(p'_1,r_1(t_0))$,   where $\bar{\mathcal A}\in\widetilde{\mathcal G^*}$.

Suppose $\mathcal A $ is an algorithm which solves the \textsc{MMUCF} problem, i.e., $\mathcal A \in \widetilde{\mathcal G^*}$. First, consider an execution of algorithm $\mathcal{A} $ for the initial configuration $\mathcal  R(t_0)$. In this case, $\lambda^*_{\mathbf{C}}=d(p_1,r_1(t_0))$. Since $\mathcal  R(t_0)$ is rotationally symmetric, algorithm $\mathcal A$ decides to move either all the robots in $\mathcal A$ at a time or the robots in $\{r_1,r_3,r_5\}$ at a time or the robots in $\{r_2,r_4,r_6\}$ at a time. Suppose at the beginning of the execution, algorithm $\mathcal A$ decides to move  $r_2,r_4$ and $r_6$ at a time towards their respective destination points $p_2$, $p_4$ and $p_6$ respectively. Since the movements of the robots are non-rigid, an adversary can stop robots before reaching their destination points. Suppose adversary stops the robots  $r_2,r_4$ and $r_6$  at points $p_2$, $p_4$ and $p_6$ respectively at time $t>t_0$ so that the initial configuration is converted to the configuration $\mathcal  R(t)$ as depicted in Figure \ref{Imp_2}(B).

The value of $\lambda^*_{\mathbf{C}}(\mathcal A, \mathcal  R(t))$ is strictly less than the value of $\lambda^*_{\mathbf{C}}(\mathcal A, \mathcal  R(t_0))$. If robots compute their final destination points in $\mathcal  R(t)$ considering the value of $\lambda^*_{\mathbf{C}}(\mathcal  R(t))$, the final destination points of the robots are then changed: for robot $r_i$, the destination point is $v'_i$, $1\le i\le6$. If robots start moving towards these points, then robots travel more than $\lambda^*_{\mathbf{C}}(\mathcal A, \mathcal  R(t_0))$ distance i.e., they violate the constraint of the \textsc{MMUCF} problem. Therefore, robots somehow need the access of the value $\lambda^*_{\mathbf{C}}(\mathcal A, \mathcal  R(t_0))$. They can do this by any one of the following ways (i) they remember the value of $\lambda^*_{\mathbf{C}}(\mathcal A, \mathcal  R(t_0))$ or (ii) they remember locations of their destination points or (iii) remember  or recompute their initial configuration $\mathcal  R(t_0)$. For options (ii) and (iii), robots need to remember some angles and/or distances in the initial configuration $\mathcal  R(t_0)$. All these methods require to store real numbers which may take  unbounded amount of memory.  Since robots operate under \textsc{Fall} model, they have constant amount of persistent memory and they can use their lights to communicate their states. The informations that the robots can carry forward from configuration $\mathcal  R(t_0)$ to $\mathcal  R(t)$ is of constant amount. Thus, they can not recompute or remember the exact value of $\lambda^*_{\mathbf{C}}(\mathcal A, \mathcal  R(t_0))$. This implies that the \textsc{MMUCF} problem is not solvable under \textsc{Fall} model.

\end{proof}

 \begin{theorem}
 \label{theorem-5}
 In the absence of persistent external lights, the \textsc{MMUCF} problem is unsolvable in the $\mathcal{FSYNC}$ model for the robots having non-rigid movements, even if robots have both axis agreement.
 \end{theorem}
 
\begin{figure*}[!ht]
    \centering
    \captionsetup[subfigure]{justification=centering}

    \begin{subfigure}[t]{0.45\textwidth}
        \centering
        \includegraphics[width=\linewidth]{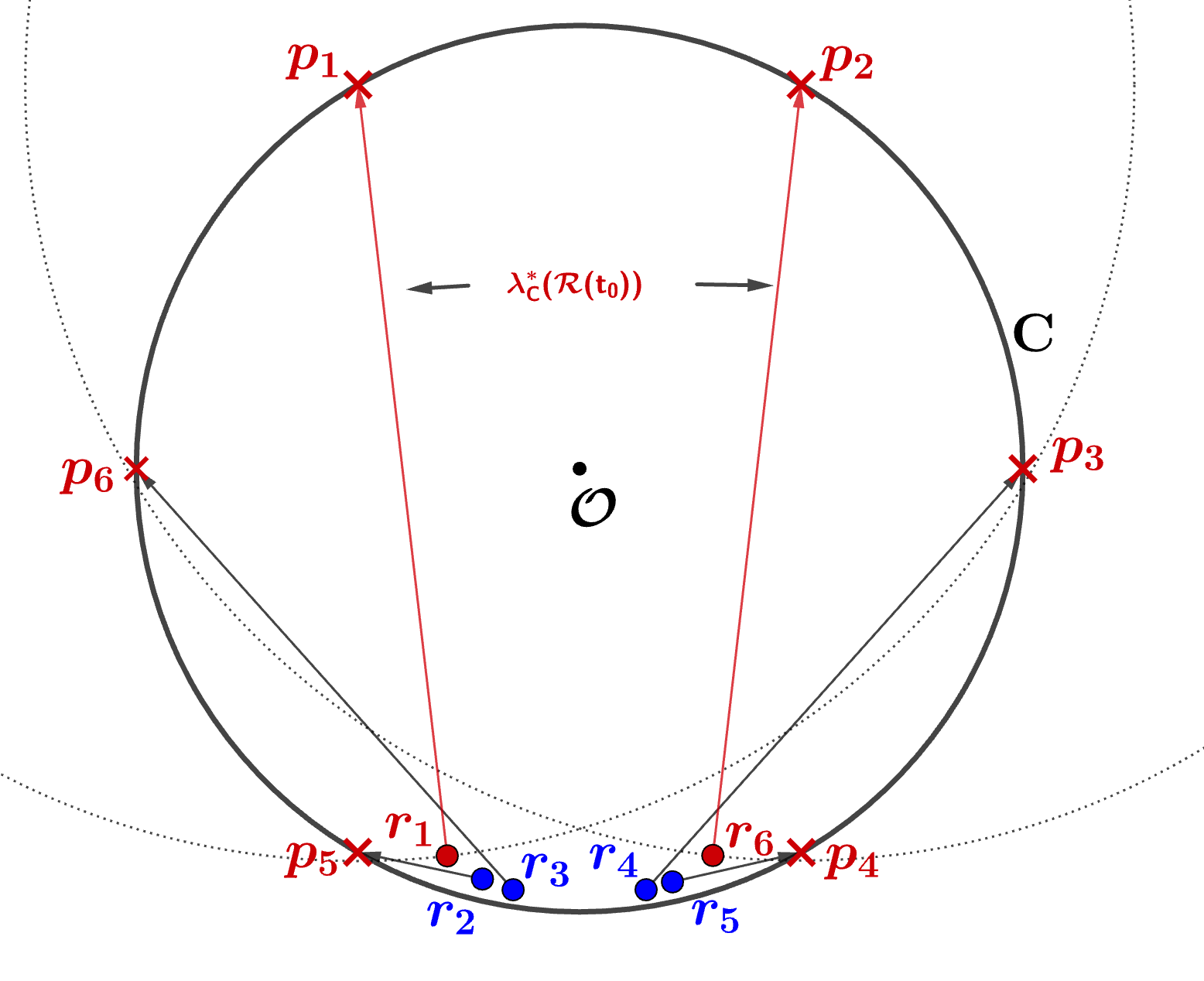}
        \caption{}
    \end{subfigure}
    \hfill
    \begin{subfigure}[t]{0.45\textwidth}
        \centering
        \includegraphics[width=\linewidth]{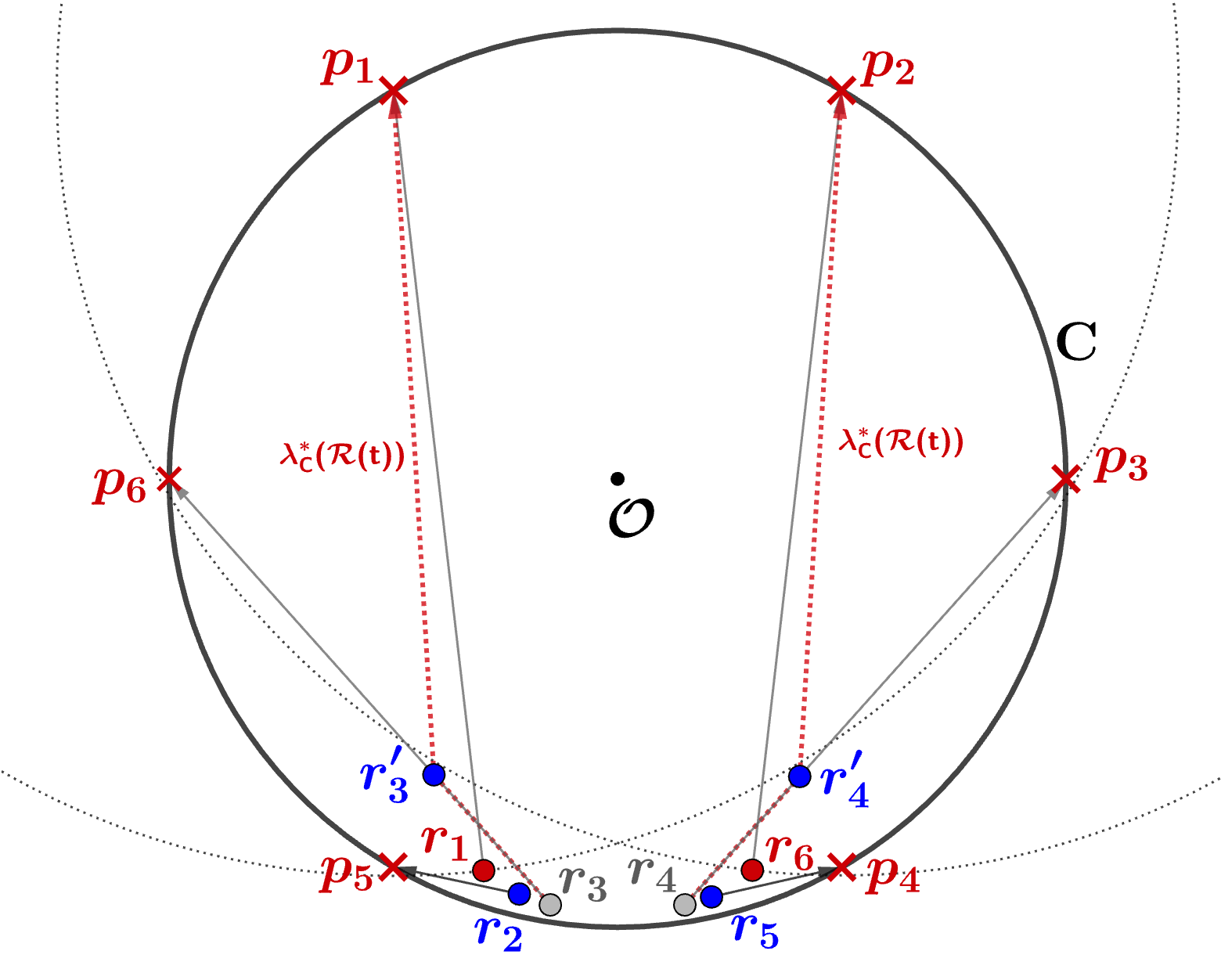}
        \caption{}
    \end{subfigure}

       \caption{An illustration in the proof of \textbf{Theorem} \ref{theorem-5}: \textbf{(A)} In this configuration of $\mathcal  R(t_0)=\{r_1(t_0), r_2(t_0), \ldots, r_6(t_0)\}$, the min-max distance is $\lambda_{\mathbf{C}}^*(\mathcal{R}(t_0)) = d(r_1, p_1)=d(r_2,p_2)$. \textbf{(B)} In this configuration of $\mathcal{R'}=\{r_1(t), r_2(t), r_3'(t),r_4'(t),r_5(t), r_6(t)\}$, the min-max distance is $\lambda_{\mathbf{C}}^*(\mathcal{R}(t)) =d(r_3'(t), p_1)=d(r_4'(t),p_2)$. }
    \label{Impo_3}
\end{figure*}

\begin{proof}
Let $\mathcal{R}=\{r_1,r_2,\ldots,r_6\}$ be a set of oblivious robots operating
under the $\mathcal{FSYNC}$ model. The robots have non-rigid movements and agree
on the directions and orientations of both local coordinate axes.

Consider an initial configuration $\mathcal{R}(t_0)$, as shown in  Figure~\ref{Impo_3}(A). Let $\mathcal P=\{p_1,p_2,\ldots,p_6\}$ be the optimal target point set on the given circle $\mathbf C$. Suppose that the
optimal value satisfies
$\lambda_{\mathbf{C}}^*(\mathcal{R}(t_0))
=
d(r_1(t_0),p_1)
=
d(r_6(t_0),p_4).$ Thus, $r_1$ and $r_6$ are critical robots in the initial configuration, since they attain the maximum distance under the optimal assignment.

Suppose that there exists an algorithm $\mathcal{A}\in\widetilde{\mathcal{G}^*}$ that solves
$\mathcal{MMUCF}$ in finite time under this model without persistent external lights. Since the robots are oblivious, they do not remember the initial configuration, the initial optimal value $\lambda_{\mathbf{C}}^*(\mathcal{R}(t_0))$, or the initial optimal assignment after a \textit{Look--Compute--Move} cycle.

To preserve the initial min--max value,i.e., $\lambda_{\mathbf{C}}^*(\mathcal{R}(t_0))$, the critical robots cannot be moved before the remaining robots are fixed at their assigned target positions. If a critical robot is moved and is stopped at an intermediate position by the non-rigid adversary, then the value of $\lambda_{\mathbf{C}}^*(\mathcal{R}(t_0))$ and the set of critical robots computed from the next snapshot may differ from those of $\mathcal{R}(t_0)$. Therefore, in any execution that preserves the initial optimal assignment, progress must first be made by non-critical robots while the initial critical robots remain stationary.

Consider an execution of $\mathcal{A}$ in which the non-critical robots $r_3$ and $r_4$ among the non-critical robots move toward their assigned target points $p_6$ and $p_3$, respectively. Since the movement is non-rigid, the adversary may stop these robots before they reach their destinations, provided that each moving robot makes the required minimum progress. Let $r_3'$ and $r_4'$ be the intermediate positions reached by $r_3$ and $r_4$, respectively, as shown in Figure~\ref{Impo_3}(B). Let $\mathcal{R}(t)$ be the resulting configuration.

The adversary chooses the intermediate positions $r_3'$ and $r_4'$ so that, in the current configuration $\mathcal{R}(t)$, the optimal value computed from the current snapshot is attained by $r_3'$ and $r_4'$. Hence, $r_3'$ and $r_4'$ become critical robots with respect to $\mathcal{R}(t)$. Let their corresponding target points in the optimal assignment of $\mathcal{R}(t)$ be $p_1$ and $p_2$, respectively. Then
$d(r_3',p_1)=d(r_4',p_2)=\lambda_{\mathbf{C}}^*(\mathcal{R}(t)).$

However, before reaching $\mathcal{R}(t)$, the robots $r_3$ and $r_4$ have already travelled the positive distances $d(r_3(t_0),r_3')$ and $d(r_4(t_0),r_4')$, respectively. Therefore, if the robots now follow the optimal assignment computed from the current configuration $\mathcal{R}(t)$, then the total distances travelled by $r_3$ and $r_4$ become $d(r_3(t_0),r_3')+d(r_3',p_1)$ and $d(r_4(t_0),r_4')+d(r_4',p_2),$  respectively.

By choosing the intermediate positions sufficiently far along the movement segments, the adversary can ensure that $d(r_3(t_0),r_3')>
\lambda_{\mathbf{C}}^*(\mathcal{R}(t_0))-\lambda_{\mathbf{C}}^*(\mathcal{R}(t))$ and $d(r_4(t_0),r_4')
>
\lambda_{\mathbf{C}}^*(\mathcal{R}(t_0))-\lambda_{\mathbf{C}}^*(\mathcal{R}(t)).$
Consequently, $d(r_3(t_0),r_3')+d(r_3',p_1)
>
\lambda_{\mathbf{C}}^*(\mathcal{R}(t_0))$ and $d(r_4(t_0),r_4')+d(r_4',p_2)
>
\lambda_{\mathbf{C}}^*(\mathcal{R}(t_0)).$ Thus, at least one robot is forced to travel more than the initial min--max optimal value $\lambda_{\mathbf{C}}^*(\mathcal{R}(t_0))$.

Now consider the information available to the robots in configuration $\mathcal{R}(t)$. Since the robots are oblivious and have no persistent external lights, they cannot determine whether $\mathcal{R}(t)$ is an initial configuration or an intermediate configuration produced after previous non-rigid movements. Therefore, the decision of the deterministic algorithm $\mathcal{A}$ depends only on the current snapshot. If the robots follow the optimal assignment computed from $\mathcal{R}(t)$, then, as shown above, some
robot travels more than $\lambda_{\mathbf{C}}^*(\mathcal{R}(t_0))$. On the other hand, preserving the initial optimal assignment requires information about $\mathcal{R}(t_0)$, which is not available to oblivious robots without persistent external lights.

This contradicts the assumption that $\mathcal{A}$ solves \textsc{MMUCF} in finite time while preserving the min--max objective. Therefore, no such algorithm exists. Hence, in the absence of persistent external lights, the \textsc{MMUCF} problem is unsolvable in the $\mathcal{FSYNC}$ model with non-rigid movements, even under both-axis agreement.
\end{proof}

\section{Algorithm \textsc{MinMaxUniformCircle()} for \textsc{MMUCF} Problem}
\label{Minmaxunifrom}

\subsection{High-level overview of the algorithm}
We present $\textsc{MinMaxUniformCircle()}$, a deterministic distributed algorithm for solving \textsc{MMUCF} under the luminous, oblivious, non-rigid $\mathcal{ASYNC}$ model. From the initial configuration $\mathcal R(t_0)$, every robot computes the optimal min--max value $\lambda_{\mathbf C}^*$, a regular target-point set $\mathcal P$, and an optimal assignment $f_{\mathrm{opt}}(t_0):\mathcal R(t_0)\rightarrow\mathcal P$ using the algorithm in~\cite{bhattacharya2009optimal}. The main challenge under obliviousness and non-rigid movements is to preserve the initially selected target-point set $\mathcal P$ and optimal assignment $f_{\mathrm{opt}}$ throughout the execution, since robots may observe intermediate configurations after partially completed moves.

The preservation mechanism depends on the critical set $\mathcal R^*(t_0)$. In a unique-critical configuration, the critical robot determines a radial reference for $\mathcal P$; if the critical robot is located at the centre $\mathcal O$, a short deterministic movement first establishes the radial reference. Multiple-critical configurations are handled according to the positions of the non-critical robots and unreserved targets relative to the collective critical region $\mathcal D^*(t_0)$. Stationary critical robots normally preserve the target-point set. In the special configuration $\mathcal K_{\mathrm{mcr}}^{\mathrm{ct}}$, with exactly two critical robots, a virtual dummy point is used to reconstruct the same target-point set during the movement of \textit{candidate robot} (see Definition~\ref{def:candidate}).

Robots act as critical, \textit{secondary critical} (see Definition~\ref{def:secondary}), or candidate robots. Critical and secondary critical robots remain stationary with lights $\lambda_{\textsc{Anchor}}$ and $\lambda^-_{\textsc{Anchor}}$, respectively, while the selected candidate moves toward its computed target point with light $\textsc{TargetTrace}$ or $\textsc{ReferenceTrace}$. A non-rigid interruption does not change its target point. Upon arrival, the robot displays the light $\textsc{TargetReached}$ and remains stationary. Candidates' robots move first, followed by secondary critical robots, and finally critical robots. This ordering preserves the target-point set and the selected assignment throughout the subsequent movements of the robots, while the fairness of the $\mathcal{ASYNC}$ scheduler guarantees finite-time completion. To ensure deterministic and consistent selections throughout the execution, whenever multiple robot positions or target points are equally eligible, all robots use the common ordering induced by the one-axis agreement to select the same robot position or target point. The overall procedure is summarized in Algorithm~\ref{alg:master}, while the detailed subroutines for the respective configuration classes are provided in Section~\ref{Assignment}.

\begin{algorithm}[ht!]
\caption{$\textsc{MinMaxUniformCircle}()$ --- executed independently
by every robot $r$ in each Look--Compute--Move cycle}
\label{alg:master}
\begin{algorithmic}[1]

\Statex \textbf{Look:}
$r$ observes the current configuration $\mathcal R(t)$ and the lights
of all robots.

\If{$r.\text{light}=\textsc{TargetReached}$}
    \State \textbf{remain stationary}; \Return
    \Comment{terminal state}
\EndIf

\State
$(\lambda_{\mathbf C}^*,\mathcal P,f_{\mathrm{opt}},
\mathcal R^*(t),\kappa)
\gets
\textsc{ClassifyAndPlan}(\mathcal R(t))$
\Comment{Algorithm~\ref{alg:classify}}

\Switch{$\kappa$}

    \Case{$\mathcal K_{\mathrm{ucr}}^{\mathrm{nc}}$}
        \State
        $\textsc{ProcUCR-NC}
        (r,\mathcal R(t),\lambda_{\mathbf C}^*,\mathcal P,
        f_{\mathrm{opt}})$
        \Comment{Algorithm~\ref{alg:ucrnc}}
    \EndCase

    \Case{$\mathcal K_{\mathrm{ucr}}^{\mathrm{c}}$}
        \State
        $\textsc{ProcUCR-C}
        (r,\mathcal R(t),\lambda_{\mathbf C}^*,\mathcal P,
        f_{\mathrm{opt}})$
        \Comment{Algorithm~\ref{alg:ucrc}}
    \EndCase

    \Case{$\mathcal K_{\mathrm{mcr}}^{\mathrm{all}}$}
        \State
        $\textsc{ProcMCR-ALL}
        (r,\mathcal R(t),\lambda_{\mathbf C}^*,\mathcal P,
        f_{\mathrm{opt}})$
        \Comment{Algorithm~\ref{alg:mcrall}}
    \EndCase

    \Case{$\mathcal K_{\mathrm{mcr}}^{\mathrm{ot}}$}
        \State
        $\textsc{ProcMCR-OT}
        (r,\mathcal R(t),\lambda_{\mathbf C}^*,\mathcal P,
        f_{\mathrm{opt}})$
        \Comment{Algorithm~\ref{alg:mcrot}}
    \EndCase

    \Case{$\mathcal K_{\mathrm{mcr}}^{\mathrm{ct}}$}
        \State
        $\textsc{ProcMCR-CT}
        (r,\mathcal R(t),\lambda_{\mathbf C}^*,\mathcal P,
        f_{\mathrm{opt}})$
        \Comment{Algorithm~\ref{alg:mcrct}}
    \EndCase

    \Case{$\mathcal K_{\mathrm{mcr}}^{\mathrm{ii}}$
    \textbf{ or }
    $\mathcal K_{\mathrm{mcr}}^{\mathrm{io}}$}
        \State
        $\textsc{ProcMCR-II-IO}
        (r,\mathcal R(t),\lambda_{\mathbf C}^*,\mathcal P,
        f_{\mathrm{opt}})$
        \Comment{Algorithm~\ref{alg:mcriiio}}
    \EndCase

\EndSwitch

\Statex \textbf{Move:}
$r$ executes the movement determined by the invoked procedure; a non-rigid interruption is handled by the corresponding persistent
light state.

\end{algorithmic}
\end{algorithm}

\subsection{Configuration Classes Based on the Critical Set}
\label{subsec:critical-classification}

Let $\mathcal R^*(t_0)$ denote the set of critical robots in the initial configuration $\mathcal  R(t_0)$. After computing the optimal min--max value $\lambda_{\mathbf C}^*$, and the corresponding target-point set $\mathcal P$, the robots classify the initial configuration according to the cardinality of $\mathcal R^*(t_0)$.

For each critical robot $r_i^*\in\mathcal R^*(t_0)$, let $p_i^*\in\mathcal P$ be its assigned critical target point. The associated closed critical disk is defined
by $\overline{\mathcal D}_i^*(t_0)
=
\left\{
x\in\mathbb R^2:
d(x,p_i^*)\leq\lambda_{\mathbf C}^*
\right\}.$ Thus, $\overline{\mathcal D}_i^*(t_0)$ is the closed disk of radius
$\lambda_{\mathbf C}^*$ centreed at $p_i^*$. The collective critical
region is defined by
\[
\mathcal D^*(t_0)
=
\bigcup_{r_i^*\in\mathcal R^*(t_0)}
\overline{\mathcal D}_i^*(t_0).
\]
Let $\mathcal P^*(t_0)
=
\left\{
p_i^*\in\mathcal P:
r_i^*\in\mathcal R^*(t_0)
\right\}$ denote the set of target points reserved for the critical robots and let $r_j\in\mathcal N(t_0)$ be a selected non-critical robot. Define the set of its possible unreserved target points by
$\mathcal P_j^{<}
=
\left\{
p\in\mathcal P_{\mathcal N}(t_0):
d\bigl(r_j(t_0),p\bigr)<\lambda_0
\right\},$ where $\mathcal P_{\mathcal N}(t_0)
= \mathcal P\setminus\mathcal P^*(t_0)$ is the set of target points not reserved for the critical robots. Accordingly, the configurations are classified into the following two main classes.

\begin{enumerate}

\item[(A)] \textbf{Unique-Critical Robot Configuration}
$(\mathcal K_{ucr})$: A configuration $\mathcal  R(t_0)$ belongs to $\mathcal K_{ucr}$ if
$|\mathcal R^*(t_0)|=1$. This class is further divided into the following two
subclasses.

\begin{enumerate}

\item[(A.1)] \textbf{Unique Non-centre-Critical Configuration}
$(\mathcal K_{ucr}^{nc})$: A configuration $\mathcal  R(t_0)$ belongs to $\mathcal K_{ucr}^{nc}$ if the
unique critical robot is not located at $\mathcal O$.

\item[(A.2)] \textbf{Unique centre-Critical Configuration}
$(\mathcal K_{ucr}^{c})$: A configuration $\mathcal  R(t_0)$ belongs to $\mathcal K_{ucr}^{c}$ if the
unique critical robot is located at the centre $\mathcal O$ of the target
circle $\mathbf C$.

\end{enumerate}

Thus,
$\mathcal K_{ucr}
=
\mathcal K_{ucr}^{nc}
\cup
\mathcal K_{ucr}^{c}$.

\item[(B)] \textbf{Multiple-Critical Robot Configurations}
$(\mathcal K_{\mathrm{mcr}})$:

A configuration $\mathcal R(t_0)$ belongs to
$\mathcal K_{\mathrm{mcr}}$ if 
$|\mathcal R^*(t_0)|\geq 2.$  We consider the family
\[
\mathcal K_{\mathrm{mcr}}^{\mathrm{out}}
=
\left\{
\mathcal R(t_0)\in\mathcal K_{\mathrm{mcr}}:
\mathcal N(t_0)\neq\emptyset
\ \text{and}\
\mathcal N(t_0)\cap\mathcal D^*(t_0)=\emptyset
\right\}.
\]
Thus, every configuration in
$\mathcal K_{\mathrm{mcr}}^{\mathrm{out}}$ contains at least one
non-critical robot, and every non-critical robot lies outside the
collective critical region $\mathcal D^*(t_0)$. We further define
\[
\mathcal K_{\mathrm{mcr}}^{\mathrm{in}}
=
\left\{
\mathcal R(t_0)\in\mathcal K_{\mathrm{mcr}}:
\mathcal N(t_0)\cap\mathcal D^*(t_0)
\neq\emptyset
\right\}.
\]
Thus, a configuration belongs to
$\mathcal K_{\mathrm{mcr}}^{\mathrm{in}}$ if at least one non-critical
robot lies inside or on the boundary of the collective critical
region.

For every
$\mathcal R(t_0)\in\mathcal K_{\mathrm{mcr}}^{\mathrm{in}}$, let
$(r_c,p_c)$ be the deterministically selected robot--target pair
satisfying
\[
(r_c,p_c)
\in
\operatorname*{arg\,min}_{
\substack{
r_j\in\mathcal N(t_0)\cap\mathcal D^*(t_0)\\
p\in\mathcal P_{\mathcal N}(t_0)
}}
d\bigl(r_j(t_0),p\bigr).
\]
Thus, $r_c$ is a non-critical robot lying in
$\mathcal D^*(t_0)$, and $p_c$ is its closest unreserved target point.
If more than one robot--target pair has the same minimum distance, the
tie is resolved using the one-axis agreement.

This family is divided
into the following three subclasses according to the positions of the
unreserved target points.

\begin{table}[t]
\caption{Classification of initial configurations based on the critical set.}
\label{tab:critical-classification}
 \centering
\begingroup
\renewcommand{\arraystretch}{1.12}
\setlength{\tabcolsep}{4pt}
\small

\begin{tabularx}{\textwidth}{
@{}
>{ \centering\arraybackslash}p{2.35cm}
>{\raggedright\arraybackslash}X
@{}
}
\toprule
\textbf{Class}
&
\textbf{Defining condition}
\\
\midrule

$\mathcal K_{\mathrm{ucr}}^{\mathrm{nc}}$
&
$|\mathcal R^*(t_0)|=1$ and the unique critical robot is not at
$\mathcal O$
\\

$\mathcal K_{\mathrm{ucr}}^{\mathrm{c}}$
&
$|\mathcal R^*(t_0)|=1$ and the unique critical robot is at
$\mathcal O$
\\

$\mathcal K_{\mathrm{mcr}}^{\mathrm{all}}$
&
$\mathcal R^*(t_0)=\mathcal R(t_0)$
\\

$\mathcal K_{\mathrm{mcr}}^{\mathrm{ot}}$
&
$\mathcal N(t_0)\cap\mathcal D^*(t_0)=\emptyset$ and
$\mathcal P_{\mathcal N}(t_0)\setminus\mathcal D^*(t_0)\neq\emptyset$
\\

$\mathcal K_{\mathrm{mcr}}^{\mathrm{ct}}$
&
$\mathcal N(t_0)\cap\mathcal D^*(t_0)=\emptyset$ and
$\mathcal P_{\mathcal N}(t_0)\subseteq\mathcal D^*(t_0)$
\\

$\mathcal K_{\mathrm{mcr}}^{\mathrm{ii}}$
&
$r_c\in\mathcal D^*(t_0)$ and
$p_c\in\mathcal D^*(t_0)$
\\

$\mathcal K_{\mathrm{mcr}}^{\mathrm{io}}$
&
$r_c\in\mathcal D^*(t_0)$ and
$p_c\notin\mathcal D^*(t_0)$
\\

\bottomrule
\end{tabularx}
\endgroup
\end{table}

\begin{enumerate}

\item[(B.1)] \textbf{All-Critical Robot Configuration}
$(\mathcal K_{\mathrm{mcr}}^{\mathrm{all}})$:

A configuration $\mathcal R(t_0)$ belongs to
$\mathcal K_{\mathrm{mcr}}^{\mathrm{all}}$ if every robot is a critical
robot with respect to the selected optimal assignment
$f_{\mathrm{opt}}$. Formally, $\mathcal K_{\mathrm{mcr}}^{\mathrm{all}}
=
\left\{
\mathcal R(t_0)\in\mathcal K_{\mathrm{mcr}}:
\mathcal R^*(t_0)=\mathcal R(t_0)
\right\}.$ In this configuration, $d\bigl(r_i(t_0),f_{\mathrm{opt}}(r_i)\bigr)
=\lambda_{\mathbf C}^*,
\ \text{for every }r_i\in\mathcal R(t_0).$ Consequently, every target point is reserved for a critical robot, and
hence $\mathcal P^*(t_0)=\mathcal P$ and $\mathcal P_{\mathcal N}(t_0)=\emptyset.$ Thus, no non-critical robot exists in this configuration.

\item[(B.2)] \textbf{Outside-Target Type}
$(\mathcal K_{\mathrm{mcr}}^{\mathrm{ot}})$:

A configuration $\mathcal R(t_0)$ belongs to
$\mathcal K_{\mathrm{mcr}}^{\mathrm{ot}}$ if there exists at least one
unreserved target point $p\in\mathcal P_{\mathcal N}(t_0)$ that lies
outside the collective critical region $\mathcal D^*(t_0)$. Formally,
\[
\mathcal K_{\mathrm{mcr}}^{\mathrm{ot}}
=
\left\{
\mathcal R(t_0)\in\mathcal K_{\mathrm{mcr}}^{\mathrm{out}}:
\exists\,p\in\mathcal P_{\mathcal N}(t_0)
\text{ such that }
p\notin\mathcal D^*(t_0)
\right\}.
\]
Thus, at least one unreserved target point lies outside
$\mathcal D^*(t_0)$. Moreover, since
$\mathcal R(t_0)\in\mathcal K_{\mathrm{mcr}}^{\mathrm{out}}$, every
non-critical robot lies outside $\mathcal D^*(t_0)$.

\item[(B.3)] \textbf{Covered-Target Type}
$(\mathcal K_{\mathrm{mcr}}^{\mathrm{ct}})$:

A configuration belongs to
$\mathcal K_{\mathrm{mcr}}^{\mathrm{ct}}$ if every unreserved target point lies inside or on the boundary of the collective critical
region. Formally, $\mathcal K_{\mathrm{mcr}}^{\mathrm{ct}}
=
\left\{
\mathcal R(t_0)\in\mathcal K_{\mathrm{mcr}}^{\mathrm{out}}:
\mathcal P_{\mathcal N}(t_0)
\subseteq
\mathcal D^*(t_0)
\right\}.$ Equivalently, $\mathcal P_{\mathcal N}(t_0)\setminus\mathcal D^*(t_0)
=\emptyset.$ Thus, every unreserved target point lies inside or on the boundary of
$\mathcal D^*(t_0)$, whereas every non-critical robot lies outside
$\mathcal D^*(t_0)$.

Therefore, $\mathcal K_{\mathrm{mcr}}^{\mathrm{out}}
=
\mathcal K_{\mathrm{mcr}}^{\mathrm{ot}}
\cup
\mathcal K_{\mathrm{mcr}}^{\mathrm{ct}}.$

\item[(B.4)] \textbf{Inner-Robot--Inner-Target Type}
$(\mathcal K_{\mathrm{mcr}}^{\mathrm{ii}})$:

A configuration $\mathcal R(t_0)$ belongs to
$\mathcal K_{\mathrm{mcr}}^{\mathrm{ii}}$ if the selected
non-critical robot $r_c$ lies in the collective critical region and
its closest unreserved target point $p_c$ also lies in the same
collective critical region. Formally, $\mathcal K_{\mathrm{mcr}}^{\mathrm{ii}}
=
\left\{
\mathcal R(t_0)\in\mathcal K_{\mathrm{mcr}}^{\mathrm{in}}:
p_c\in\mathcal D^*(t_0)
\right\}.$ Thus, both $r_c(t_0)$ and its closest unreserved target point $p_c$
lie inside or on the boundary of $\mathcal D^*(t_0)$.

\item[(B.5)] \textbf{Inner-Robot--Outer-Target Type}
$(\mathcal K_{\mathrm{mcr}}^{\mathrm{io}})$:

A configuration $\mathcal R(t_0)$ belongs to
$\mathcal K_{\mathrm{mcr}}^{\mathrm{io}}$ if the selected
non-critical robot $r_c$ lies in the collective critical region but
its closest unreserved target point $p_c$ lies outside this region.
Formally, $\mathcal K_{\mathrm{mcr}}^{\mathrm{io}}
=
\left\{
\mathcal R(t_0)\in\mathcal K_{\mathrm{mcr}}^{\mathrm{in}}:
p_c\notin\mathcal D^*(t_0)
\right\}.$ Thus, $r_c(t_0)$ lies inside or on the boundary of
$\mathcal D^*(t_0)$, whereas its closest unreserved target point
$p_c$ lies outside $\mathcal D^*(t_0)$.

Therefore,
$\mathcal K_{\mathrm{mcr}}^{\mathrm{in}}
=
\mathcal K_{\mathrm{mcr}}^{\mathrm{ii}}
\cup
\mathcal K_{\mathrm{mcr}}^{\mathrm{io}}.$
\end{enumerate}
 So we have
$\mathcal K_{\mathrm{mcr}}
=
\mathcal K_{\mathrm{mcr}}^{\mathrm{all}}
\cup
\mathcal K_{\mathrm{mcr}}^{\mathrm{ot}}
\cup
\mathcal K_{\mathrm{mcr}}^{\mathrm{ct}}
\cup
\mathcal K_{\mathrm{mcr}}^{\mathrm{ii}}
\cup
\mathcal K_{\mathrm{mcr}}^{\mathrm{io}}.$

\end{enumerate}

Thus, $\mathcal K_{\mathrm{ucr}}
=
\mathcal K_{\mathrm{ucr}}^{\mathrm{nc}}
\cup
\mathcal K_{\mathrm{ucr}}^{\mathrm{c}},$ and $\mathcal K_{\mathrm{mcr}}
=
\mathcal K_{\mathrm{mcr}}^{\mathrm{all}}
\cup
\mathcal K_{\mathrm{mcr}}^{\mathrm{ot}}
\cup
\mathcal K_{\mathrm{mcr}}^{\mathrm{ct}}
\cup
\mathcal K_{\mathrm{mcr}}^{\mathrm{ii}}
\cup
\mathcal K_{\mathrm{mcr}}^{\mathrm{io}}.$ Table~\ref{tab:critical-classification} summarizes the configuration
classes used by the algorithm.

\subsection{Configuration Classification and Planning}
\label{app:classify-plan}

The following procedure summarizes how a robot determines the configuration class from its current snapshot.

\begin{algorithm}[ht]
\caption{$\textsc{ClassifyAndPlan}(\mathcal R(t))$}
\label{alg:classify}
\begin{algorithmic}[1]
\State Compute $\lambda_{\mathbf C}^*(\bar{\mathcal A},\mathcal R(t))$ and the target-point set $\mathcal P$ using~\cite{bhattacharya2009optimal,chen2015optimal}
\State Compute the selected optimal assignment $f_{\mathrm{opt}}$ (one-axis agreement breaks ties among multiple optima)
\State $\mathcal R^*(t)\gets \{r_i\in\mathcal R(t): d(r_i(t),f_{\mathrm{opt}}(r_i))=\lambda_{\mathbf C}^*\}$
\If{$|\mathcal R^*(t)|=1$}
    \State \Return $\mathcal K_{ucr}^{nc}$ \textbf{if} the critical robot $\neq \mathcal O$ \textbf{else} $\mathcal K_{ucr}^{c}$
\Else
    \If{$\mathcal R^*(t)=\mathcal R(t)$} \Return $\mathcal K_{mcr}^{\mathrm{all}}$ \EndIf
    \State Compute $\mathcal D^*(t)$, $\mathcal P_{\mathcal N}(t)$
    \If{$\exists\, p\in\mathcal P_{\mathcal N}(t): p\notin\mathcal D^*(t)$} \Return $\mathcal K_{mcr}^{\mathrm{ot}}$ \EndIf
    \If{$\mathcal P_{\mathcal N}(t)\subseteq\mathcal D^*(t)$ \textbf{and} $\mathcal N(t)\cap\mathcal D^*(t)=\emptyset$} \Return $\mathcal K_{mcr}^{\mathrm{ct}}$ \EndIf
    \State Select $(r_c,p_c)\gets\arg\min_{r_j\in\mathcal N(t)\cap\mathcal D^*(t),\,p\in\mathcal P_{\mathcal N}(t)} d(r_j(t),p)$
    \State \Return $\mathcal K_{mcr}^{\mathrm{ii}}$ \textbf{if} $p_c\in\mathcal D^*(t)$ \textbf{else} $\mathcal K_{mcr}^{\mathrm{io}}$
\EndIf
\end{algorithmic}
\end{algorithm}

\subsection{States of Robots}

To maintain the invariance of the target point set and the optimal assignment $f_{opt}$ during
the execution, robots use visible lights to encode their current states. Let
$\mathcal{X}$ denote the finite set of possible light states. In our algorithm,

$\mathcal X
=
\left\{
\textsc{Off}, 
\lambda_{\textsc{Anchor}},
\lambda^-_{\textsc{Anchor}},
\textsc{TargetTrace},
\textsc{ReferenceTrace},
\textsc{TargetReached}
\right\}.$

The meaning of each state is as follows.

\begin{itemize}

\item $\textsc{Off}$ indicates that the robot has not yet performed its first Look--Compute--Move cycle.

\item $\lambda_{\textsc{Anchor}}$ indicates that the robot is \textit{critical}
and currently acts as an anchor for the optimal min--max value
$\lambda_{\mathbf C}^*$ and the target-point set. A robot displaying
$\lambda_{\textsc{Anchor}}$ remains stationary until the selected
candidate robot reaches its target point and displays
$\textsc{TargetReached}$.

\item $\lambda^-_{\textsc{Anchor}}$ indicates that the robot is the
\textit{secondary critical} robot. Although it is not itself a critical robot, its movement is deferred until every currently selected \textit{candidate} robot reaches its target point and displays $\textsc{TargetReached}$.

\item $\textsc{TargetTrace}$ indicates that the \textit{candidate} robot is moving directly toward its assigned target point. If its movement is interrupted because of non-rigidity, it retains $\textsc{TargetTrace}$ and continues toward the same target point in
its next activation. Upon reaching the target, it changes its light
to $\textsc{TargetReached}$.

\item $\textsc{ReferenceTrace}$ indicates the special movement used
for a configuration in $\mathcal K_{\mathrm{mcr}}^{\mathrm{ct}}$. The robot moves toward its assigned target point while reconstructing the original target-point
set from the dummy point and the stationary critical reference
robots. It retains $\textsc{ReferenceTrace}$ throughout every
interrupted movement and changes its light to
$\textsc{TargetReached}$ upon reaching the target.

\item $\textsc{TargetReached}$ indicates that the robot has reached
its assigned target point. A robot displaying
$\textsc{TargetReached}$ remains stationary and never moves again.

\end{itemize}

\subsection{Computation of the Assignment}
\label{Assignment}
There may exist more than one optimal assignment of the robot positions to the points in $\mathcal P$. We have the following observation.


To solve the \textsc{MMUCF} problem, robots need to agree on a common optimal assignment $f_{opt}$ to the points in $\mathcal P$. Moreover, in the asynchronous setting with \emph{non-rigid} movements,
robots may observe intermediate configurations. Therefore, the initially
computed assignment $f_{\mathrm{opt}}$ must be preserved throughout the
execution; in particular, no robot may change its assigned target on the
basis of an intermediate snapshot. In this section, we describe an algorithm to output such an assignment. We assume that, in the initial configuration, no robot lies on the circle $\mathbf C$. However, during the execution of the algorithm, robots move to the points on $\mathbf C$. Our algorithm uses axis-only agreement to compute such an assignment $f_{\mathrm{opt}}$ (as stated in Theorem~\ref{theorem-3}), particularly in symmetric configurations or when a candidate robot has two target points at the same minimum distance.

The computation of the optimal assignment $f_{\mathrm{opt}}$ depends on
the structure of the critical robot set in the initial configuration
$\mathcal R(t_0)$. The configuration may contain either a unique
critical robot or multiple critical robots. The unique-critical case is
further divided according to whether the critical robot is located at
the centre $\mathcal O$ of the target circle $\mathbf C$, since a robot
located at $\mathcal O$ does not initially determine a unique radial
target point. Accordingly, the computation of $f_{\mathrm{opt}}$ is
divided into the following three cases: (1) the configuration contains a unique critical robot that is not located at $\mathcal O$; (2) the configuration contains a unique critical robot located at $\mathcal O$; and (3) the configuration contains multiple critical robots. The algorithms presented in~\cite{bhattacharya2009optimal}
and~\cite{chen2015optimal} are used to compute $\lambda_{\mathbf C}^*(\mathcal R(t_0))$ and identify the corresponding critical set, which determines the assignment procedure to be applied.\\

  \begin{figure}[ht]
     \centering
     \vspace{-0.3 cm}
     
     \includegraphics[width=0.4\textwidth]{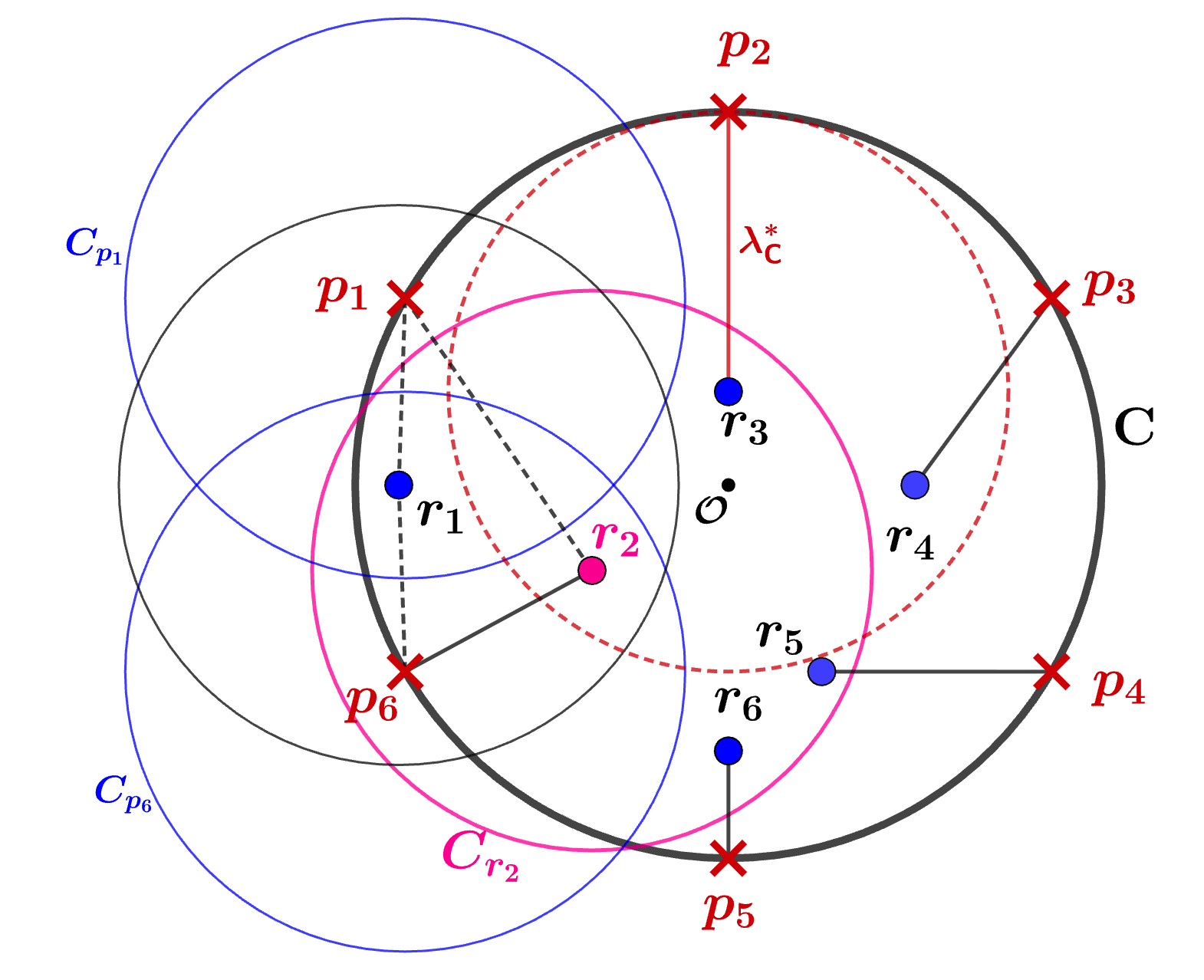}
          \caption{Illustration of secondary critical robot in a configuration of six robots $r_1,r_2,\dots,r_6.$ In this configuration, the critical distance is $\lambda_{\mathbf{C}}^*=d(r_3,p_2)$ and $r_2$ is a \textit{secondary critical robot} with \textit{secondary critical target position} $p_6$.}
     \label{second}
     \end{figure}

\begin{definition}[Secondary Critical Robot]
\label{def:secondary}
Let $\mathcal P^*$ be the set of critical target positions. A robot
$r_s \in \mathcal{R}(t)$ is called a \emph{secondary critical robot} if
$r_s$ is not a critical robot and the circle $\mathcal{C}_k$, centreed at
$r_s$ with radius $\lambda_{\mathbf{C}}^*$, contains exactly one target position
$p_s \in \mathcal P \setminus \mathcal P^*$ such that $d(r_s,p_s)<\lambda_{\mathbf{C}}^*$. The target position $p_s$ is called the \emph{secondary critical target
position} corresponding to $r_s$.
\end{definition}

The purpose of introducing secondary critical robots is to preserve the invariance of the optimal assignment $f_{opt}$ under the asynchronous scheduler with non-rigid movements. Critical robots preserve the value of $\lambda_{\mathbf{C}}^*$, but they may not be sufficient to preserve the optimal assignment $f_{opt}$ of non-critical robots. During the movement of a candidate robot, the set of feasible target positions for other robots may change, and this may force some robots to travel more than
$\lambda_{\mathbf{C}}^*$. A secondary critical robot is therefore kept fixed temporarily so that its uniquely determined target position remains reserved. This prevents the optimal assignment $f_{opt}$ from changing before the candidate robots settle. Figure~\ref{second} illustrates this situation. The robot $r_3$ is a critical robot because its assigned target point $p_2$ is at distance $\lambda_{\mathbf{C}}^*$ from $r_3$; hence $p_2\in\mathcal P^*$. The robot $r_2$ is not critical, but the circle $C_{r_2}$ centreed at $r_2$ with radius $\lambda_{\mathbf{C}}^*$ contains exactly one target point in $\mathcal P\setminus\mathcal P^*$, namely $p_6$. Therefore, $r_2$ is a secondary critical robot and $p_6$ is its secondary critical target position. If $r_2$ is allowed to move too early or if its target position is not kept fixed, then another candidate robot may be assigned to $p_6$, which can change the optimal assignment $f_{opt}$ and may violate the min--max bound.\\

\begin{definition}[Candidate Robot]
\label{def:candidate}
Let $\mathcal R_{\mathrm{sec}}(t)$ denote the set of secondary critical
robots, and let $\mathcal P_{\mathrm{sec}}(t)$ denote their reserved
target points. Define $\mathcal N_{\mathrm{can}}(t)
=
\mathcal N(t)\setminus\mathcal R_{\mathrm{sec}}(t),
\ \mathcal P_{\mathrm{can}}(t)
=
\mathcal P\setminus
\bigl(\mathcal P^*\cup\mathcal P_{\mathrm{sec}}(t)\bigr).$ A robot $r_j\in\mathcal N_{\mathrm{can}}(t)$ is called a
\emph{candidate robot} if there exists a target point
$p_j\in\mathcal P_{\mathrm{can}}(t)$ such that $(r_j,p_j)
\in
\operatorname*{arg\,min}_{
\substack{
r\in\mathcal N_{\mathrm{can}}(t)\\
p\in\mathcal P_{\mathrm{can}}(t)
}}
d\bigl(r(t),p\bigr).$ The pair $(r_j,p_j)$ is called the corresponding
\emph{candidate robot--target pair}. This selection rule is used in
every unique-critical and multiple-critical configuration.
\end{definition}

\subsubsection{Common Light-Transition Patterns}
\label{app:common-light-transitions}

\[
\begin{array}{rcl}
\text{Candidate robot}
&:&
\text{\textsc{Off}}
\longrightarrow
\text{\textsc{TargetTrace}}
\longrightarrow
\text{\textsc{TargetReached}}
\\[4pt]
\text{Secondary critical robot}
&:&
\text{\textsc{Off}}
\longrightarrow
\lambda^-_{\textsc{Anchor}}
\longrightarrow
\text{\textsc{TargetTrace}}
\longrightarrow
\text{\textsc{TargetReached}}
\\[4pt]
\text{Critical robot}
&:&
\text{\textsc{Off}}
\longrightarrow
\lambda_{\textsc{Anchor}}
\longrightarrow
\text{\textsc{TargetTrace}}
\longrightarrow
\text{\textsc{TargetReached}}.
\end{array}
\]

\subsubsection{Assignment Computation for $\mathcal K_{ucr}$}

In this subsection, we describe the computation of the optimal
assignment $f_{\mathrm{opt}}$ for configurations containing exactly one
critical robot. According to the classification above, $\mathcal K_{ucr}
=
\mathcal K_{ucr}^{nc}
\cup
\mathcal K_{ucr}^{c}.$ Therefore, the assignment computation is considered separately according to whether the unique critical robot is located away from or at the centre $\mathcal O$ of the target circle $\mathbf C$. We first consider configurations in $\mathcal K_{ucr}^{nc}$. The following lemma identifies the target point associated with the unique non-centre critical robot and thereby determines the radial direction used to fix the target-point set.

\begin{lemma}
\label{lem:unique-critical-nearest}
Let $\mathbf C$ be the target circle with centre $\mathcal O$, and let $r_c\neq \mathcal O$ be the unique critical robot in an optimal assignment $f_{opt}$. Let $p_c$ be the intersection point of the circumference of $\mathbf C$ with the ray starting from $\mathcal O$ and passing through $r_c$. Then, $p_c\in \mathcal P$, where $\mathcal P$ denotes the optimal target point set.
\end{lemma}

\begin{proof}
Let $p_c\in \mathcal P$ be the target point assigned to $r_c$. Since $r_c$ is critical, $ d(r_c,p_c)=\lambda_{\mathbf{C}}^*$. Because $r_c$ is the unique critical robot, no other robot attains the value $\lambda_{\mathbf{C}}^*$. By \textbf{Lemma~2} in~\cite{chen2015optimal}, if an optimal solution attains the min--max value at a sensor $A_i$, then either its assigned target point is the closest boundary point $X_i$, or there exists another sensor attaining the same value. Since the second case is impossible here, we obtain $p_c=p_c$. Hence,
$p_c\in \mathcal P$. Further, since $r_c\neq O$, the point $p_c$ is exactly the nearest intersection point of the line through $\mathcal O$ and $r_c$ with the circumference.
\end{proof}

\begin{observation}[Angular rigidity of a unique-critical configuration]
\label{obs:unique-critical-angular-rigidity}

Let $r_c\neq\mathcal O$ be the unique critical robot, and let $p_c$ be
its reserved target point with $d(r_c(t_0),p_c)=\lambda_{\mathbf C}^*$. The point $p_c$, together with the angular separation $2\pi/n$, uniquely determines the regular target polygon $\mathrm{Poly}_{\mathcal P}$. If this polygon is rotated about
$\mathcal O$ through any sufficiently small nonzero angle, either
clockwise or counter-clockwise, the distance from $r_c$ to its nearest
target point becomes strictly greater than $\lambda_{\mathbf C}^*$.
Hence, the min--max cost of the rotated target-point set is strictly
greater than $\lambda_{\mathbf C}^*$.
\end{observation}

\noindent\textbf{Case-I (Configuration
$\mathcal K_{\mathrm{ucr}}^{\mathrm{nc}}$).}
In this scenario, the initial asymmetric configuration $\mathcal  R(t_0)$ about the line $\mathcal L$ contains exactly one critical robot, not lying at centre $\mathcal O$, say $r_c$ and let this line $\mathcal L$ intersect the circumference of the given circle $\mathbf C$ at $p_1$ (nearest). Fix an ordering~\cite{sruti2016}, say $\mathcal{F}$ of the points in $\mathcal P$. Let $(p_1 , p_2,\ldots, p_n)$ be the sequence of points in $\mathcal P$ according to $\mathcal{F}$. Let $\mathcal P_{C}$ be the set of points of the critical robot target positions in $\mathcal P$. 

    \begin{figure}[ht]
     \centering
     \vspace{-0.3 cm}
     
     \includegraphics[width=0.5\textwidth]{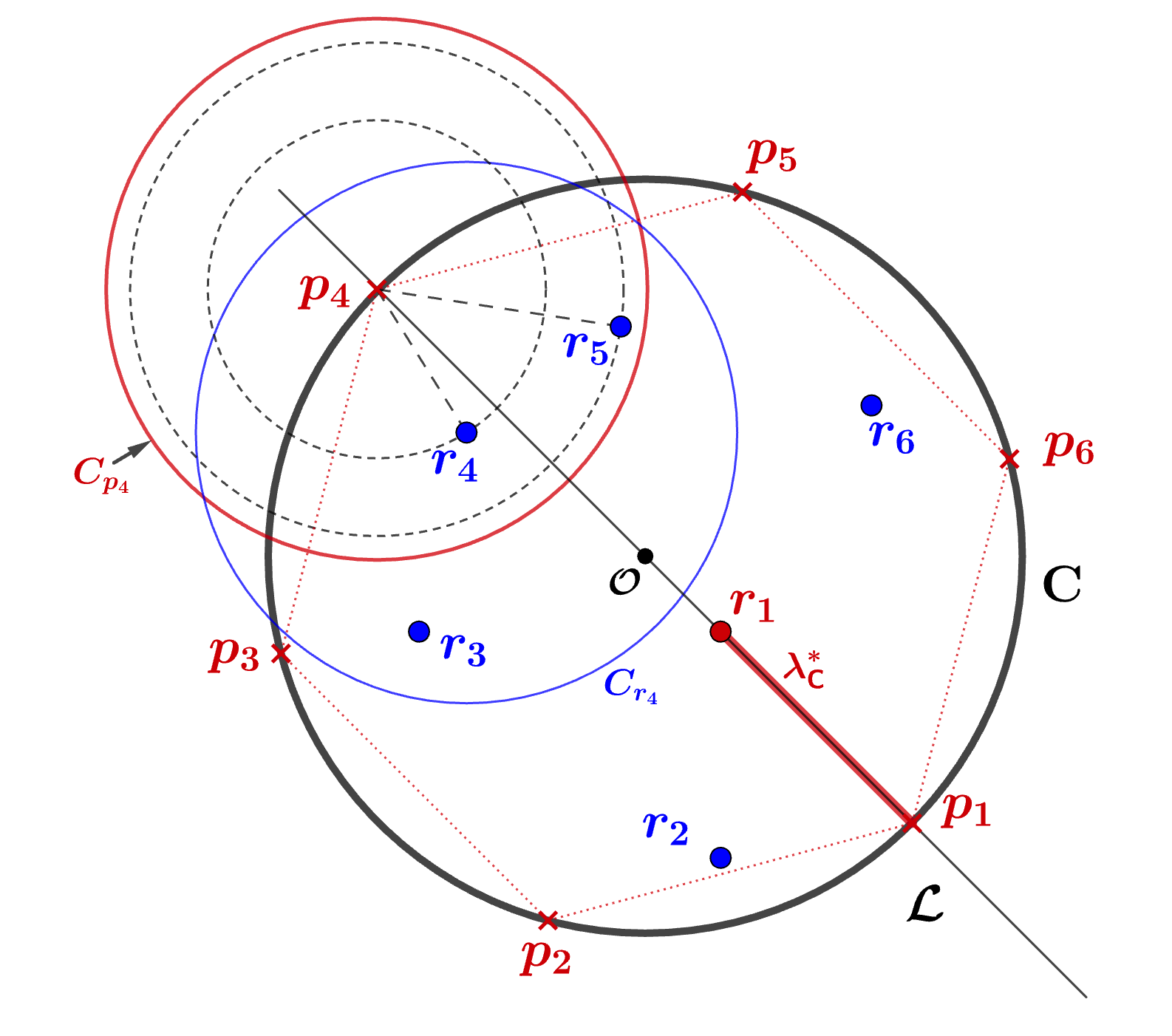}
     
     \caption{Illustration of a configuration of six robots $r_1,r_2,\dots,r_6$. Here, $r_c=r_1$, $\mathcal P^*={p_1}$, and the optimal assignment $f_{opt}(r_i)=p_i$, $i=1,2,\ldots,6$.}
     \label{single}
     \end{figure}
    
Let $p_{i} \in \mathcal P \setminus \mathcal P_{C}$. Draw a circle $C_{p_{i}}$ of radius $\lambda^{*}_{C}$ centreed at $p_{i}$. Let $\mathcal{R}_{p_{i}}(t)$ be the set of all robot positions lying on or inside the circle $C_{p_{i}}$. If the circle $C_{p_{i}}$ contains a single robot $r_i$, i.e., $|\mathcal{R}_{p_{i}}| = 1$, then the target position of $r_i$ is $p_i$, i.e., $r_i \rightarrow p_i$. If the circle $C_{p_{i}}$ contains more than one robot, i.e., $|\mathcal{R}_{p_{i}}| > 1$, then the target positions of all robots in $\mathcal{R}_{p_{i}}(t)$ are computed as follows: Define an order $\mathcal{H}$ on $\mathcal{R}_{p_{i}}(t)$ based on the distance between the robots in $\mathcal{R}_{p_{i}}(t)$ and the vertex $p_i$. Let robots $r_j$ and $r_k$ belong to $\mathcal{R}_{p_{i}}(t)$ with $d(r_j, p_i) < d(r_k, p_i)$. If either of these two robots is not a secondary critical robot for vertex $p_i$, then the target position of $r_j$ is $p_i$, i.e., $r_j \rightarrow p_i$. If either of these two robots is a secondary critical robot, say $r_k$ for vertex $p_i$, then the target position of $r_k$ is $p_i$, i.e., $r_k \rightarrow p_i$ (See Figure~\ref{single}). If $r_j, r_k \in \mathcal{R}_{p_{i}}(t)$ with $d(r_j, p_i) = d(r_k, p_i)$ (where the distance is closest to $p_i$), then to break this tie, we use the agreement on the $Y$-axis as follows: let $\mathcal{B}$ be the set of robots that are at the same distance from $p_{i}$ and lie on the circle $\mathcal{C}_{b}$ centreed at $p_{i}$. Let $c_{1}$ and $c_{2}$ be the two intersection points between $\mathcal{C}_{b}$ and $\mathbf C$. Consider the set $\mathcal P \cup \{c_{1}, c_{2}\}$. Fix an ordering $\mathcal{J}$ of the points in $\mathcal P \cup \{c_{1}, c_{2}\}$. Without loss of generality, suppose that $c_{1}$ has a higher order in $\mathcal{J}$ than $c_{2}$. Order the robots in $\mathcal{B}$ according to their distances from $c_{1}$ to break the ties and obtain the original ordering of the robot positions in $\mathcal{R}_{p_{i}}(t)$. Let $(r_j,p_j)$ be the candidate robot--target pair selected according to the candidate-selection rule defined above.

    \begin{figure}[h]
     \centering
     \vspace{-0.3 cm}
     
     \includegraphics[width=0.4\textwidth]{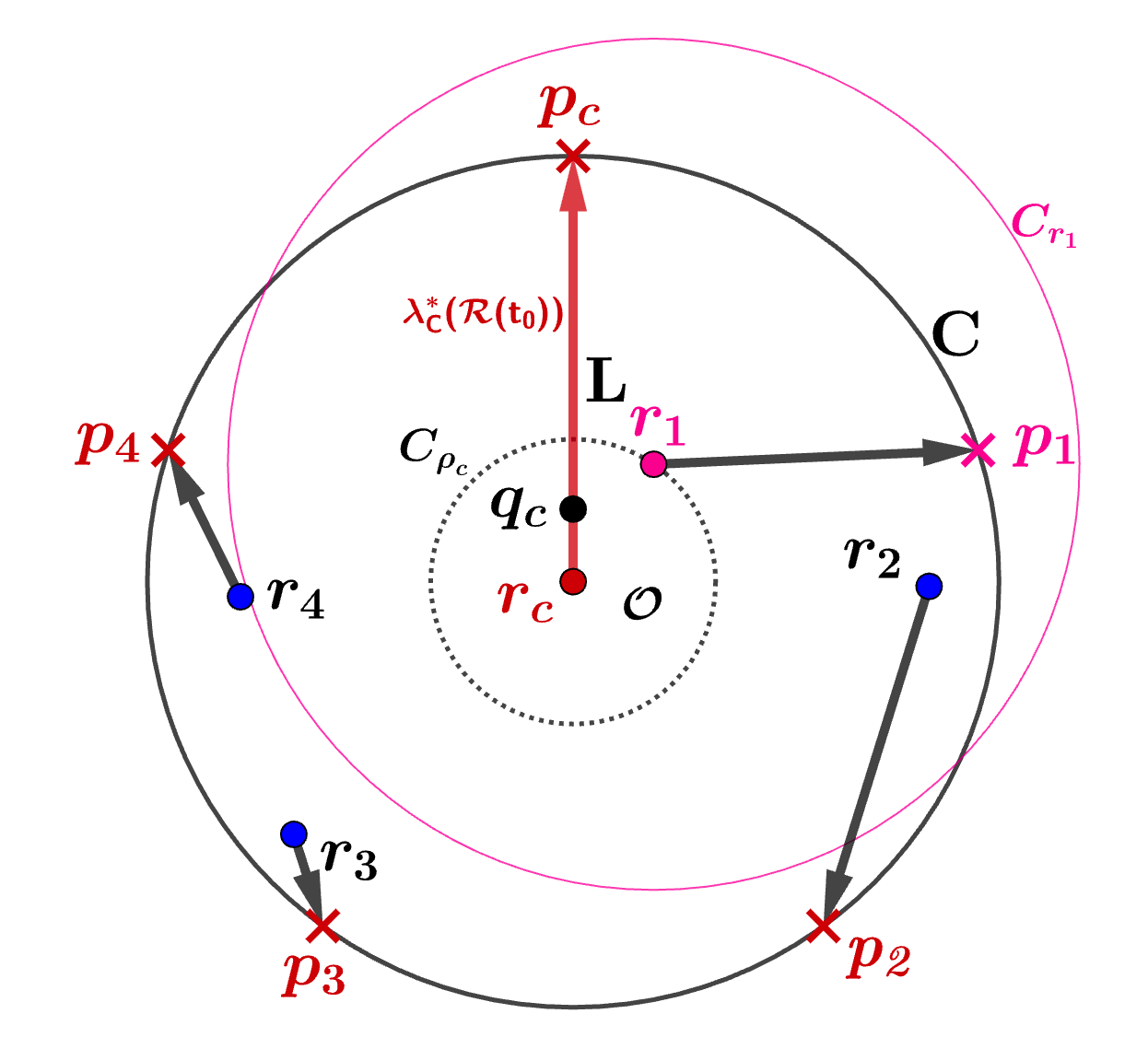}
     
\caption{Illustration of the target-point construction for a centre-critical robot $r_c$ for $\mathcal{R}=\{r_c,r_1,r_2,r_3,r_4\}$. The maximum angular sector $S_{23}$ determines the bisector $\mathbf L$, along which the auxiliary point $q_c$ and the target point $p_c$ are selected.}
     \label{Centre_3}
     \end{figure}

During the execution of the above assignment, each robot changes its
light according to its assigned role. When the unique critical robot
$r_c$ is activated, it changes its light from $\textsc{Off}$ to
$\lambda_{\textsc{Anchor}}$ and remains stationary.
For a target point $p_i$, if
$|\mathcal R_{p_i}(t)|=1$, the selected candidate robot is $r_i$.
If $|\mathcal R_{p_i}(t)|>1$ and no secondary critical robot is
selected for $p_i$, the candidate robot selected by the ordering
$\mathcal H$ is $r_j$. In either case, the selected candidate changes
its light according to $\textsc{Off}
\longrightarrow
\textsc{TargetTrace}$ and moves toward $p_i$.
If $r_k$ is identified as the secondary critical robot for $p_i$, it
changes its light from $\textsc{Off}$ to
$\lambda^-_{\textsc{Anchor}}$ and remains stationary while the candidate robots move. Every other robot also remains stationary
whenever a robot displays $\textsc{TargetTrace}$.
A robot displaying $\textsc{TargetTrace}$ retains this light whenever
its movement is interrupted and continues toward the same target point
in its next activation. Upon reaching its target, it changes its light
to $\textsc{TargetReached}$ and remains stationary.

After all candidate robots display $\textsc{TargetReached}$, the
secondary critical robot $r_k$, if one exists, changes its light from
$\lambda^-_{\textsc{Anchor}}$ to $\textsc{TargetTrace}$ and moves
toward its assigned target point $p_i$. Finally, after every non-critical robot has reached its target, the unique critical robot
$r_c$ changes its light from $\lambda_{\textsc{Anchor}}$ to
$\textsc{TargetTrace}$ and moves toward its assigned critical target
point $p_1$. Each of these robots changes its light to
$\textsc{TargetReached}$ upon reaching its target. The robots follow the common light-transition pattern given in~\ref{app:common-light-transitions}.

\begin{algorithm}[ht]
\caption{$\textsc{ProcUCR-NC}(r,\mathcal R(t),\lambda_{\mathbf C}^*,\mathcal P,f_{\mathrm{opt}})$ --- configuration $\mathcal K_{ucr}^{nc}$}
\label{alg:ucrnc}
\begin{algorithmic}[1]
\State $r_c \gets$ the unique critical robot; $p_c\gets f_{\mathrm{opt}}(r_c)$
\State Determine $r_s$ (secondary critical robot, Def.~\ref{def:secondary}), if any, and its target $p_s$
\State $(r_j,p_j)\gets$ candidate--target pair selected by the candidate-selection rule
\Statex \Comment{--- Light transitions ---}
\If{$r=r_c$}
    \If{$r.\text{light}=\textsc{Off}$} \State $r.\text{light}\gets\lambda_{\textsc{Anchor}}$; \textbf{remain stationary} \Return \EndIf
    \If{every $r_i\neq r_c$ displays $\textsc{TargetReached}$} \State $r.\text{light}\gets\textsc{TargetTrace}$; \textbf{move toward } $p_c$
    \Else\ \textbf{remain stationary} \EndIf
\ElsIf{$r=r_s$}
    \If{$r.\text{light}=\textsc{Off}$} \State $r.\text{light}\gets\lambda^-_{\textsc{Anchor}}$; \textbf{remain stationary} \Return \EndIf
    \If{every candidate robot displays $\textsc{TargetReached}$} \State $r.\text{light}\gets\textsc{TargetTrace}$; \textbf{move toward } $p_s$
    \Else\ \textbf{remain stationary} \EndIf
\ElsIf{$r=r_j$}
    \State $r.\text{light}\gets\textsc{TargetTrace}$; \textbf{move toward } $p_j$
\Else
    \State \textbf{remain stationary} \Comment{not yet eligible}
\EndIf
\If{$r$ reaches its assigned target} $r.\text{light}\gets\textsc{TargetReached}$ \EndIf
\end{algorithmic}
\end{algorithm}


\noindent\textbf{Case-II (Configuration
$\mathcal K_{\mathrm{ucr}}^{\mathrm{c}}$).}
We consider a configuration having a unique critical robot, denoted by $r_c$, which is located at the centre $\mathcal O$ of the given circle $\mathbf C$; that is, $\mathcal R^*(t_0)=\{r_c\}$ and $r_c(t_0)=\mathcal O.$ The configuration may be asymmetric or symmetric; in the symmetric case, $r_c$ breaks the symetricity. Since $r_c$ is located at the centre of $\mathbf C$, every point on the circumference of $\mathbf C$ is at the same distance from $r_c(t_0)$. Therefore, the target point corresponding to $r_c$ is not unique initially (See Observation~\ref{obs-centre-ambiguity}). Hence, before fixing the target point set, a deterministic direction must be selected so that the ambiguity in choosing the target point of $r_c$ is removed.

Let  $\rho_c
=
\min_{\substack{
r\in\mathcal R(t_0)\\
r\neq r_c
}}
d\bigl(r_c(t_0),r(t_0)\bigr).$
Thus, $\rho_c$ denotes the minimum distance from the centre-critical robot $r_c$ to any other robot in the initial configuration. Consider the circle $C_{\rho_c}$ centreed at $\mathcal O$ with radius $\rho_c$. The robot $r_c$ considers the rays from $\mathcal O$ through all non-central robots in their cyclic order. For every pair of consecutive rays, including the last and the first rays, it considers the angular sector containing no other non-central robot in its interior. It then
selects a sector of maximum angular measure, denoted by $S_{ij}$. The
same set of sectors is obtained whether the cyclic order is traversed
clockwise or counterclockwise. If several maximum sectors exist, the
tie is resolved using the one-axis agreement. The angle bisector ray from $\mathcal{O}$, say $\mathbf{L}$ of this sector $S_{ij}$ is the direction of the first movement of $r_c$. We choose an auxiliary point $q_c$ on $\mathbf{L}$ such that $0<d(\mathcal O,q_c)<\rho_c$ (see Figure~\ref{Centre_3}). In particular, $q_c$ can be chosen with $d(\mathcal O,q_c)=\rho_c/2$. This choice ensures that $q_c$ lies strictly closer to $\mathcal O$ than to any other robot position. Hence, when the direction from $\mathcal O$ to $q_c$ is used for fixing the target point of the centre-critical robot, no other robot becomes critical, and the original centre-critical robot $r_c$ remains the critical robot. Moreover, $q_c$ is chosen so that it does not lie on any line passing through $\mathcal O$ and another robot position in $\mathcal  R(t_0)\setminus\{r_c\}$. This avoids choosing a direction already determined by another robot position and fixes a deterministic direction from $\mathcal O$ for computing the target point set. Let $p_{c}$ be the intersection point of $\mathbf L$ with the circumference of $c$. The point $p_{c}$ is chosen as the target point corresponding to the centre-critical robot $r_c$. Once $p_{c}$ is fixed, the remaining target points are placed on $c$ by maintaining equal arc distance $\frac{2\pi h}{n}$ between consecutive target points, where $h$ is the radius of $c$. Thus, the ambiguity caused by the centre-critical robot is removed, and the target point set is uniquely determined.
The remaining robots compute their assignments with respect to this fixed target point set by following the optimal assignment $f_{opt}$ procedure described above. In particular, no new target point set is recomputed for the remaining robots; they use the target point set determined by the line $\mathbf L$. Therefore, the auxiliary choice of $q_c$ is used only to fix a deterministic direction for computing the target point set. It does not create any new critical robot and does not violate the optimal min--max value $\lambda_{\mathbf{C}}^*(\bar{\mathcal A},\mathcal  R(t_0))$.

\begin{algorithm}[ht]
\caption{$\textsc{ProcUCR-C}(r,\mathcal R(t),\lambda_{\mathbf C}^*,\mathcal P,f_{\mathrm{opt}})$ --- configuration $\mathcal K_{ucr}^{c}$}
\label{alg:ucrc}
\begin{algorithmic}[1]
\State $r_c\gets$ the unique critical robot located at $\mathcal O$
\If{$r=r_c$ \textbf{and} $r.\text{light}=\textsc{Off}$}
    \State Compute $\rho_c$, the maximum empty angular sector, its bisector ray $\mathbf L$, and $q_c$ with $d(\mathcal O,q_c)=\delta<\min\{\rho_c,\gamma_c\}$ \Comment{Lemma~\ref{lem:centre-invariant-reduction}}
    \State $r.\text{light}\gets\lambda_{\textsc{Anchor}}$; \textbf{move toward } $q_c$
    \State \Return \Comment{every other robot remains stationary during this step}
\EndIf
\If{$r=r_c$ \textbf{and} $r$ has reached $q_c$}
    \State $\mathcal R(t)$ now belongs to $\mathcal K_{ucr}^{nc}$ with critical target $p_c$ fixed by ray $\mathbf L$
    \State $\textsc{ProcUCR-NC}(r,\mathcal R(t),\lambda_{\mathbf C}^*,\mathcal P,f_{\mathrm{opt}})$ \Comment{delegate; $r_c$ retains $\lambda_{\textsc{Anchor}}$ until eligible}
\ElsIf{$r\neq r_c$}
    \State $\textsc{ProcUCR-NC}(r,\mathcal R(t),\lambda_{\mathbf C}^*,\mathcal P,f_{\mathrm{opt}})$ \Comment{remaining robots follow the $\mathcal K_{ucr}^{nc}$ procedure}
\EndIf
\end{algorithmic}
\end{algorithm}

During the execution of this case, the centre-critical robot $r_c$ first moves slightly away from the centre $\mathcal O$ in the direction of $\mathbf{L}$. When $r_c$ is activated, it changes its light according to $\textsc{Off}
\longrightarrow \lambda_{\textsc{Anchor}}.$ Let $q_c$ be the point on the selected ray $\mathbf L$ satisfying $0
<
d(\mathcal O,q_c)
=
\delta
<
\rho_c.$ Robot $r_c$ moves from $\mathcal O$ toward $q_c$ while retaining
$\lambda_{\textsc{Anchor}}$. Every other robot remains stationary
during this movement. Since the destination is at distance $\delta$,
$r_c$ reaches $q_c$ in this movement even under the non-rigid model.
By the choice $\delta<\rho_c$, robot $r_c$ remains the unique critical
robot after reaching $q_c$. Since $q_c\neq\mathcal O$, the resulting
configuration belongs to
$\mathcal K_{\mathrm{ucr}}^{\mathrm{nc}}$. The ray from $\mathcal O$ through $q_c$ fixes the critical target point
$p_c$ and consequently the complete target-point set. Moreover, by
Corollary~\ref{cor:centre-travel-bound}, this initial movement does not
violate the original optimal min--max travel bound. Therefore, the
configuration changes according to
$\mathcal K_{ucr}^{c}\longrightarrow\mathcal K_{ucr}^{nc}$. After this transformation, the candidate robot--target pair $(r_j,p_j)$ is selected according to the candidate-selection rule
defined above. The robots then follow the procedure described for
$\mathcal K_{\mathrm{ucr}}^{\mathrm{nc}}$. In particular, $r_c$
retains $\lambda_{\textsc{Anchor}}$ while the candidate and secondary critical robots complete their movements. Finally, after the remaining $n-1$ robots have reached their respective target points and display $\textsc{TargetReached}$, robot $r_c$ changes its light from $\lambda_{\textsc{Anchor}}$ to $\textsc{TargetTrace}$ and moves toward $p_c$. Upon reaching $p_c$, it changes its light to $\textsc{TargetReached}$ and remains stationary. The robots follow the common light-transition pattern given in~\ref{app:common-light-transitions}.


\begin{observation}[Angular balance of multiple critical robots]
\label{obs:critical-angular-balance}
Let $\mathcal R(t)\in\mathcal K_{\mathrm{mcr}}$, and let $\theta^*$ be
the selected optimal rotation. For any sufficiently small rotation
that preserves the critical robot--target associations, decreasing the
distance of one critical robot forces the distance of another critical
robot to remain at least $\lambda_{\mathbf C}^*$. Otherwise, the
min--max value could be reduced below $\lambda_{\mathbf C}^*$,
contradicting the optimality of $\theta^*$. Hence, while the critical
robots remain stationary, an angular rotation cannot reduce the
min--max value determined by them.
\end{observation}

\begin{observation}[First movement in an all-critical configuration]
\label{obs:all-critical-first-move}

Let
$\mathcal R(t_0)\in\mathcal K_{\mathrm{mcr}}^{\mathrm{all}}$.
Suppose that the selected critical robots move toward their assigned
targets, while the remaining critical robots stay stationary. After
any positive movement, each selected robot becomes non-critical since
its remaining distance is strictly smaller than
$\lambda_{\mathbf C}^*$. The stationary critical robots continue to
preserve $\lambda_{\mathbf C}^*$ and the target-point set, regardless
of whether the initial symmetry is preserved or broken.
\end{observation}

\subsubsection{Assignment Computation for $\mathcal K_{mcr}$}
In a unique critical robot configuration, the line passing through the centre
$\mathcal O$ and the critical robot determines a target point on
$\mathbf C$. This point anchors the remaining target points at angular
intervals of $2\pi/n$, making the target-point set invariant.

In a multiple-critical configuration, no single critical robot provides
such a unique anchor. Moreover, robot movements may change the observed
critical structure, making it difficult to reconstruct the same
target-point set from an intermediate configuration. Therefore, the
robots use their lights to preserve the required reference information
and consistently recompute the original target-point set without
changing its orientation or the assignment. Once at least one robot
reaches its assigned target point, that occupied target point serves as
an anchor, and the complete target-point set can be reconstructed by
placing the remaining points on $\mathbf C$ at angular intervals of
$2\pi/n$. According to the classification, $\mathcal K_{mcr}$, we consider these two types separately and describe how the target-point set and assignment are reconstructed after at least one robot has moved. Before describing the assignment computation for the different multiple-critical configurations, we recall the observations concerning critical robots given in Appendix.

\noindent\textbf{Case-I (Configuration
$\mathcal K_{\mathrm{mcr}}^{\mathrm{all}}$).}
In this configuration, every robot is critical and every target point
is reserved under the selected optimal assignment. If the
configuration is rotationally symmetric, the robots use the common
ordering induced by the one-axis agreement and select a critical robot
having the highest $y$-coordinate with respect to the agreed axis. If
several symmetry-equivalent robots have the same highest position, the
highest-priority robot is selected. If they remain indistinguishable,
they may be selected simultaneously, provided that they have distinct
assigned target points.

\begin{algorithm}[ht]
\caption{$\textsc{ProcMCR-ALL}(r,\mathcal R(t),\lambda_{\mathbf C}^*,\mathcal P,f_{\mathrm{opt}})$ --- configuration $\mathcal K_{mcr}^{\mathrm{all}}$}
\label{alg:mcrall}
\begin{algorithmic}[1]
\If{$r.\text{light}=\textsc{Off}$}
    \State $r.\text{light}\gets\lambda_{\textsc{Anchor}}$; \textbf{remain stationary} \Return
\EndIf
\State $r_j\gets$ selected critical robot(s) by one-axis agreement / symmetry-breaking rule (Case-I, \S\ \ref{Assignment})
\If{$r=r_j$ (or $r\in$ selected symmetric set)}
    \State $r.\text{light}\gets\textsc{TargetTrace}$; \textbf{move toward } $f_{\mathrm{opt}}(r)$
    \If{$r$ reaches its target} $r.\text{light}\gets\textsc{TargetReached}$ \EndIf
\Else
    \State \textbf{remain stationary}, retaining $\lambda_{\textsc{Anchor}}$
    \Comment{once $r_j$ moves positively, $\mathcal R(t)\to\mathcal K_{mcr}^{\mathrm{ii}}$; subsequent cycles call $\textsc{ProcMCR-II-IO}$}
\EndIf
\end{algorithmic}
\end{algorithm}

Let $r_j$ be a selected critical robot and let
$p_j^*=f_{\mathrm{opt}}(r_j)$ be its assigned critical target point.
Robot $r_j$ moves along the segment
$[r_j(t_0),p_j^*]$, while every unselected critical robot remains
stationary. Due to non-rigid movement, $r_j$ may be stopped before
reaching $p_j^*$. In that case, it retains the same target point and
continues toward it in every subsequent activation. No additional
robot is allowed to initiate a movement until every currently selected
robot reaches its assigned target point.
Since $r_j$ is initially critical, $d\bigl(r_j(t_0),p_j^*\bigr)
=
\lambda_{\mathbf C}^*.$ Hence, $r_j(t_0)$ lies on the boundary of its closed critical disk $\overline{\mathcal D}_j^*(t_0)$. After any positive movement toward
$p_j^*$, at some time $t_1>t_0$, $d\bigl(r_j(t_1),p_j^*\bigr)
<
\lambda_{\mathbf C}^*,$ and therefore $r_j$ becomes non-critical with respect to the retained assignment.
Moreover, since $\overline{\mathcal D}_j^*(t_0)$ is convex, $[r_j(t_0),p_j^*]
\subseteq
\overline{\mathcal D}_j^*(t_0)
\subseteq
\mathcal D^*(t_0).$ Consequently, $r_j(t_1)\in\mathcal D^*(t_0)
\ \text{and}\ p_j^*\in\mathcal D^*(t_0).$ Once $r_j$ becomes non-critical, its target point $p_j^*$ becomes an
unreserved target point. Thus, the resulting configuration belongs to
$\mathcal K_{\mathrm{mcr}}^{\mathrm{ii}}$.
By Observation~\ref{obs:all-critical-first-move}, the stationary
critical robots continue to preserve
$\lambda_{\mathbf C}^*$ and the target-point set, even if unequal
non-rigid movements break the initial symmetry. Therefore, every
moving robot retains its previously assigned target point. After this transition, every subsequent candidate robot--target pair
is selected according to the candidate-selection rule defined above,
and the robots follow the procedure for
$\mathcal K_{\mathrm{mcr}}^{\mathrm{ii}}$..
During the execution of this case, every critical robot changes its
light from $\textsc{Off}$ to
$\lambda_{\textsc{Anchor}}$ when it is activated. Based on the
selection rule described above, the selected critical robot $r_j$
changes its light according to $\lambda_{\textsc{Anchor}}
\longrightarrow
\textsc{TargetTrace}$ and starts moving toward its previously assigned critical target point $p_j^*$. Every unselected critical robot retains $\lambda_{\textsc{Anchor}}$ and remains stationary.
If several symmetry-equivalent critical robots are selected simultaneously, each of them changes its light to $\textsc{TargetTrace}$ and moves toward a distinct assigned target point. Every other robot remains stationary while these lights are
visible.

Due to non-rigid movement, a selected robot $r_j$ may be stopped before
reaching $p_j^*$. In that case, it retains
$\textsc{TargetTrace}$ and continues toward the same target point in
each subsequent activation. After making any positive movement,
$r_j$ becomes non-critical, but its light and assigned target point
remain unchanged. The stationary robots continue to display
$\lambda_{\textsc{Anchor}}$ and preserve
$\lambda_{\mathbf C}^*$ and the target-point set.
Thus, after the transition $\mathcal K_{\mathrm{mcr}}^{\mathrm{all}}
\longrightarrow
\mathcal K_{\mathrm{mcr}}^{\mathrm{ii}},$ the moving robot $r_j$ continues to execute the candidate-robot procedure for
$\mathcal K_{\mathrm{mcr}}^{\mathrm{ii}}$. Upon reaching $p_j^*$, it
changes its light according to $\textsc{TargetTrace}
\longrightarrow
\textsc{TargetReached}$ and remains stationary.
After every initially selected robot displays
$\textsc{TargetReached}$, the remaining robots follow the light and
movement procedure specified for
$\mathcal K_{\mathrm{mcr}}^{\mathrm{ii}}$. Each remaining critical
robot retains $\lambda_{\textsc{Anchor}}$ until it becomes eligible
to move. It then changes its light to $\textsc{TargetTrace}$, moves
toward its assigned target point, and changes its light to
$\textsc{TargetReached}$ upon reaching it.
The light sequences are summarized as follows:

\[
\begin{array}{rcl}
\text{Initially selected critical robot }r_j
&:&
\begin{aligned}[t]
\text{\textsc{Off}}
&\longrightarrow \lambda_{\textsc{Anchor}}
\longrightarrow \text{\textsc{TargetTrace}}\\
&\longrightarrow \text{\textsc{TargetReached}}
\end{aligned}
\\[6pt]
\text{Remaining critical robot }r_i^*
&:&
\begin{aligned}[t]
\text{\textsc{Off}}
&\longrightarrow \lambda_{\textsc{Anchor}}
\longrightarrow \text{\textsc{TargetTrace}}\\
&\longrightarrow \text{\textsc{TargetReached}}.
\end{aligned}
\end{array}
\]

\noindent\textbf{Case-II (Configuration
$\mathcal K_{\mathrm{mcr}}^{\mathrm{ot}}$).}
In this configuration, at least one unreserved target point lies
outside the collective critical region. Let
$\mathcal P_{\mathrm{out}}(t_0)
=
\left\{
p\in\mathcal P_{\mathcal N}(t_0):
p\notin\mathcal D^*(t_0)
\right\}$ denote the set of such target points. Let $(r_j,p_j)$ be selected according to the candidate-selection rule defined above, with the candidate target set restricted to
$\mathcal P_{\mathrm{out}}(t_0)$. Thus, $r_j$ is a non-critical robot closest to an unreserved target
point $p_j$ lying outside $\mathcal D^*(t_0)$. If more than one
robot--target pair has the same minimum distance, the tie is resolved
using the one-axis agreement.
The selected robot $r_j$ moves along the segment
$[r_j(t_0),p_j]$, while every other robot remains stationary. If its
movement is interrupted because of non-rigidity, it retains the same
target point and continues toward $p_j$ in each subsequent activation.
No other robot is allowed to initiate a movement until $r_j$ reaches
$p_j$. Once $r_j$ reaches $p_j$, the remaining assignment procedure
is continued using the fixed target-point set.

\begin{algorithm}[ht]
\caption{$\textsc{ProcMCR-OT}(r,\mathcal R(t),\lambda_{\mathbf C}^*,\mathcal P,f_{\mathrm{opt}})$ --- configuration $\mathcal K_{mcr}^{\mathrm{ot}}$}
\label{alg:mcrot}
\begin{algorithmic}[1]
\State $\mathcal P_{\mathrm{out}}(t)\gets\{p\in\mathcal P_{\mathcal N}(t): p\notin\mathcal D^*(t)\}$
\State $(r_j,p_j)\gets$ candidate--target pair selected by the candidate-selection rule restricted to $\mathcal P_{\mathrm{out}}(t)$
\State $r_s\gets$ secondary critical robot, if any
\If{$r\in\mathcal R^*(t)$}
    \State $r.\text{light}\gets\lambda_{\textsc{Anchor}}$ (\textbf{if} $\textsc{Off}$); \textbf{remain stationary until every non-critical robot reaches its target}
    \State \textbf{then} $r.\text{light}\gets\textsc{TargetTrace}$; \textbf{move toward } $f_{\mathrm{opt}}(r)$
\ElsIf{$r=r_s$}
    \State $r.\text{light}\gets\lambda^-_{\textsc{Anchor}}$ (\textbf{if} $\textsc{Off}$); \textbf{remain stationary until all candidates reach targets}
    \State \textbf{then} $r.\text{light}\gets\textsc{TargetTrace}$; \textbf{move toward } its reserved target
\ElsIf{$r=r_j$}
    \State $r.\text{light}\gets\textsc{TargetTrace}$; \textbf{move toward } $p_j$
\Else
    \State \textbf{remain stationary}
\EndIf
\If{$r$ reaches its assigned target} $r.\text{light}\gets\textsc{TargetReached}$ \EndIf
\end{algorithmic}
\end{algorithm}

During the execution of this case, every critical robot
$r_i^*\in\mathcal R^*(t_0)$ changes its light from $\textsc{Off}$ to
$\lambda_{\textsc{Anchor}}$ when it is activated and remains
stationary. A secondary critical robot $r_s$, if one exists, changes
its light from $\textsc{Off}$ to
$\lambda^-_{\textsc{Anchor}}$ and also remains stationary.
The selected candidate robot $r_j$ changes its light according to $\textsc{Off}
\longrightarrow
\textsc{TargetTrace}$ and moves toward the selected target point $p_j$. Every other robot
remains stationary while $r_j$ displays
$\textsc{TargetTrace}$. If its movement is interrupted because of
non-rigidity, $r_j$ retains the same light and continues toward $p_j$
in each subsequent activation.
Upon reaching $p_j$, robot $r_j$ changes its light according to $\textsc{TargetTrace}
\longrightarrow
\textsc{TargetReached}$ and remains stationary. Its position then provides a direct reference
for the fixed target-point set.
After every candidate robot displays $\textsc{TargetReached}$, the
secondary critical robot $r_s$, if one exists, changes its light from
$\lambda^-_{\textsc{Anchor}}$ to
$\textsc{TargetTrace}$ and moves toward its assigned target point.
After all non-critical robots have reached their respective targets,
the critical robots become eligible to move toward their assigned
critical target points. Each critical robot changes its light from
$\lambda_{\textsc{Anchor}}$ to $\textsc{TargetTrace}$ when it
becomes eligible. Every robot changes its light to
$\textsc{TargetReached}$ upon reaching its target and remains
stationary thereafter. The robots follow the common light-transition pattern given in~\ref{app:common-light-transitions}.

\noindent\textbf{Case-III (Configuration
$\mathcal K_{\mathrm{mcr}}^{\mathrm{ct}}$).}
Let $\lambda_0
=
\lambda_{\mathbf C}^*
\bigl(\bar{\mathcal A},\mathcal R(t_0)\bigr),$ and let
$(\mathcal P_0,f_{\mathrm{opt}})$ be the deterministically selected
optimal target-point set and assignment at time $t_0$. In this
configuration, every non-critical robot lies outside
$\mathcal D^*(t_0)$, whereas every unreserved target point lies inside
or on its boundary. Therefore, a candidate robot may enter
$\mathcal D^*(t_0)$ while moving toward its assigned target point. Due
to non-rigid movements and the oblivious nature of the robots, the
main difficulty is to ensure that $\mathcal P_0$ and
$f_{\mathrm{opt}}$ remain invariant throughout all intermediate ASYNC
configurations.
Before any candidate robot moves, every critical robot changes its
light according to $\text{\textsc{Off}}
\longrightarrow
\lambda_{\textsc{Anchor}}$ and remains stationary. A
secondary critical robot, if one exists, changes its light according to $\text{\textsc{Off}}
\longrightarrow
\lambda^-_{\textsc{Anchor}}$ and also remains stationary.
Let $(r_j,p_j)$ be the candidate robot--target pair selected according
to the candidate-selection rule defined above. If a
symmetry-equivalent pair or orbit is selected, let $S$ denote the set
of all such selected candidate robots. If a robot displaying $\textsc{TargetReached}$ is visible, its
position is given priority over every critical or dummy reference for
reconstructing the target-point set.

\begin{algorithm}[ht]
\caption{$\textsc{ProcMCR-CT}(r,\mathcal R(t),\lambda_{\mathbf C}^*,\mathcal P,f_{\mathrm{opt}})$ --- configuration $\mathcal K_{mcr}^{\mathrm{ct}}$}
\label{alg:mcrct}
\begin{algorithmic}[1]
\State $(r_j,p_j)\gets$ candidate--target pair selected by the candidate-selection rule; $S\gets$ symmetry-equivalent selected set
\If{$r\in\mathcal R^*(t)$}
    \State $r.\text{light}\gets\lambda_{\textsc{Anchor}}$ (\textbf{if} $\textsc{Off}$); \textbf{remain stationary until all non-critical robots settle}
\ElsIf{$r\in S$ \textbf{and} $|\mathcal R^*(t)|>2$} \Comment{Subcase III(a)}
    \State $r.\text{light}\gets\textsc{TargetTrace}$; \textbf{move toward } $p_j=f_{\mathrm{opt}}(r)$ \Comment{reference $=$ stationary critical robots}
\ElsIf{$r\in S$ \textbf{and} $|\mathcal R^*(t)|=2$} \Comment{Subcase III(b)}
    \State Recompute $m_{ab},\ell_{ab},H^+,z^+$, dummy point $w=\tfrac{m_{ab}+z^+}{2}$
    \State $r.\text{light}\gets\textsc{ReferenceTrace}$; \textbf{move toward } $p_j=f_{\mathrm{opt}}(r)$, using $w$ to reconstruct $\mathcal P$
\ElsIf{$r=r_s$ (secondary critical)}
    \State $r.\text{light}\gets\lambda^-_{\textsc{Anchor}}$ (\textbf{if} $\textsc{Off}$); \textbf{remain stationary until all candidates reach $\textsc{TargetReached}$}
\Else
    \State \textbf{remain stationary} (\textsc{Off}) \Comment{eligible only after the first \textsc{TargetReached}}
\EndIf
\If{$r$ reaches its assigned target}
    \State $r.\text{light}\gets\textsc{TargetReached}$ \Comment{its position now anchors $\mathcal P$ for all later robots}
\EndIf
\end{algorithmic}
\end{algorithm}

\medskip
\noindent
\textbf{Subcase-III(a): More Than Two Critical Robots.} This subcase applies when $|\mathcal R^*(t_0)|>2.$ The positions and
$\lambda_{\textsc{Anchor}}$ lights of the stationary critical robots
determine the original target-point set $\mathcal P_0$. Therefore, no
dummy point or additional geometric reference is required.
Every selected candidate $r_j\in S$ changes its light according to
$\text{\textsc{Off}}
\longrightarrow
\text{\textsc{TargetTrace}}$ and moves toward its assigned target
point $p_j=f_{\mathrm{opt}}(r_j)$.
For every $r_j\in S$, once it makes a positive movement toward $p_j$,
at any time $t>t_0$ after such movement, we have
$d\bigl(r_j(t),p_j\bigr)
=
d\bigl(r_j(t_0),p_j\bigr)
-
d\bigl(r_j(t_0),r_j(t)\bigr)
<
d\bigl(r_j(t_0),p_j\bigr)
<
\lambda_0.$ Thus, every moving candidate remains non-critical
throughout its movement. Since all critical robots remain stationary,
Observation~\ref{obs:all-critical-first-move} ensures that
$\lambda_{\mathbf C}^*
\bigl(\bar{\mathcal A},\mathcal R(t)\bigr)
=
\lambda_0,\ \mathcal P(t)=\mathcal P_0.$ If the movement of a
candidate is interrupted because of non-rigidity, it retains
\textsc{TargetTrace} and continues toward the same target point in its
next activation. 

\begin{figure}[ht!]
     \centering
     \vspace{-0.3 cm}
     
     \includegraphics[width=0.5\textwidth]{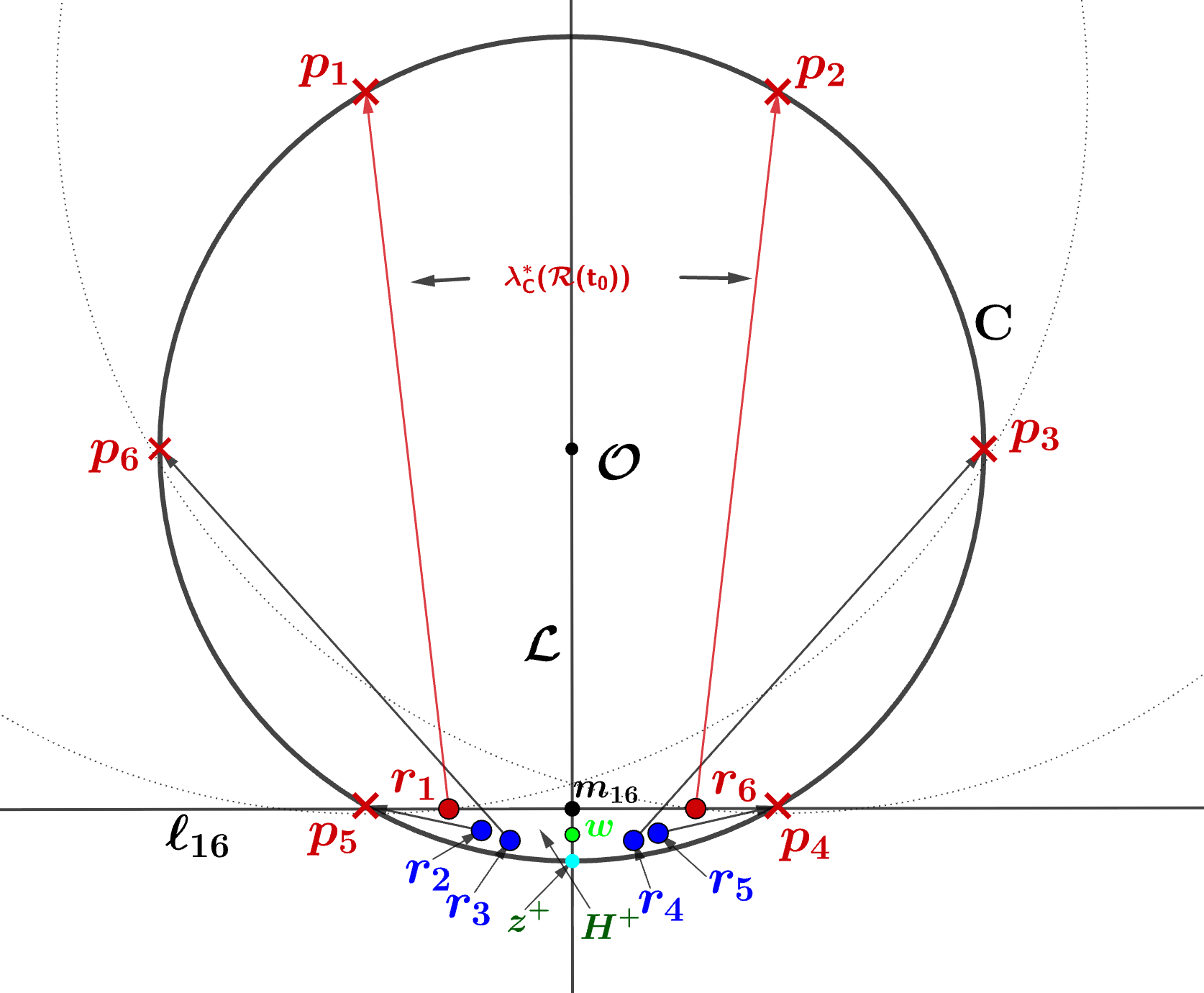}
     
\caption{A six-robot configuration with exactly two critical robots, $r_a^*=r_1$ and $r_b^*=r_6$, and critical distance $\lambda_{\mathbf C}^*$. The midpoint $m_{16}$, selected region $H^+$, and dummy point $w$ determine the invariant target-point set $\mathcal P_0$.}
    \label{multi}
     \end{figure}

\medskip
\noindent
\textbf{Subcase-III(b): Exactly Two Critical Robots.} This subcase applies when $|\mathcal R^*(t_0)|=2.$ Let $r_a^*$ and
$r_b^*$ be the two critical robots. As illustrated in
Figure~\ref{multi}, let
$m_{ab}=\frac{r_a^*(t_0)+r_b^*(t_0)}{2}$ be the midpoint of the segment
joining them. Define $\ell_{ab}
=
\overleftrightarrow{r_a^*(t_0)r_b^*(t_0)},
\ \mathscr L
=
\overleftrightarrow{\mathcal O m_{ab}},$ where
$m_{ab}\neq\mathcal O$.

The line $\ell_{ab}$ divides the disk bounded by $\mathbf C$ into two
regions. The robots select the region containing the largest number
of stationary robots. The selected candidates in $S$ are excluded
from this count because their positions may change during the movement
phase. Denote the selected region by $H^+$. Note that
$\mathcal D^*(t_0)\cap H^+=\emptyset$.

Let $z^+$ be the intersection of $\mathscr L$ with the circumference
of $\mathbf C$ lying in $H^+$. The virtual dummy point is defined by
$w=\frac{m_{ab}+z^+}{2}.$ Thus, $w$ lies on $\mathscr L$ in $H^+$,
and every candidate computes the same virtual point. The dummy point
provides a common geometric reference from which every robot
reconstructs the invariant target-point set $\mathcal P_0$, consisting
of $n$ distinct points regularly spaced on the circumference of
$\mathbf C$. Since it is neither a robot position nor a target point,
it does not participate in the computation of $\lambda_0$ or
$f_{\mathrm{opt}}$.
Each selected candidate $r_j\in S$ changes its light according to
$\text{\textsc{Off}}
\longrightarrow
\text{\textsc{ReferenceTrace}}$ and moves toward its assigned target
point $p_j=f_{\mathrm{opt}}(r_j)$.
If the movement of a candidate is interrupted, it retains
\textsc{ReferenceTrace}. Since the two critical robots and all
uncounted stationary robots remain fixed, the candidate recomputes the
same region $H^+$, dummy point $w$, and target-point set
$\mathcal P_0$ before continuing toward the same target point.
For every $r_j\in S$ and every time $t$ during its movement,
$d\bigl(r_j(t),p_j\bigr)
<
d\bigl(r_j(t_0),p_j\bigr)
<
\lambda_0.$ Since the two critical robots remain stationary, they
preserve the optimal min--max value and the original target-point set:
$\lambda_{\mathbf C}^*
\bigl(\bar{\mathcal A},\mathcal R(t)\bigr)
=
\lambda_0,\ \mathcal P(t)=\mathcal P_0.$
While any initially selected candidate in $S$ displays
\textsc{TargetTrace} or \textsc{ReferenceTrace}, every non-critical
robot outside $S$ retains \textsc{Off} and remains stationary. The
critical robots and the secondary critical robot retain their
respective anchor lights and also remain stationary. In the subcase
$|\mathcal R^*(t_0)|=2$, no robot displays \textsc{TargetTrace} during
the \textsc{ReferenceTrace} phase. Therefore, no new candidate is
selected until every initially selected candidate in $S$ reaches its
assigned target point.
Upon reaching $p_j$, a selected candidate changes its light according
to
\[
\begin{cases}
\text{\textsc{TargetTrace}}
\longrightarrow
\text{\textsc{TargetReached}},
&
\text{if }|\mathcal R^*(t_0)|>2,\\[4pt]
\text{\textsc{ReferenceTrace}}
\longrightarrow
\text{\textsc{TargetReached}},
&
\text{if }|\mathcal R^*(t_0)|=2.
\end{cases}
\]
The position $p_j$ of the first robot displaying
\textsc{TargetReached} fixes the target-point set $\mathcal P_0$. The
remaining $n-1$ target points are placed successively on the
circumference of $\mathbf C$ at an equal angular separation of
$2\pi/n$ from $p_j$. Equivalently, consecutive target points are
separated by an arc distance of $2\pi h/n$, where $h$ is the radius of
$\mathbf C$.
After all initially selected candidates have reached their targets,
the remaining robots compute their assignments with respect to the
fixed target-point set $\mathcal P_0$, excluding the occupied target
points. The remaining candidate robots move toward their assigned
targets using $\text{\textsc{Off}}
\longrightarrow
\text{\textsc{TargetTrace}}
\longrightarrow
\text{\textsc{TargetReached}}.$
After all remaining candidate robots display
$\textsc{TargetReached}$, the secondary critical robot, if one exists,
moves toward its assigned target point. Finally, the critical robots
move toward their respective critical target points according to the procedure described for $\mathcal K_{\mathrm{ucr}}^{\mathrm{nc}}$. 
The light sequences are summarized as follows:
\[
\begin{array}{rcl}
\text{Initial candidate, }|\mathcal R^*(t_0)|>2
&:&
\text{\textsc{Off}}
\longrightarrow
\text{\textsc{TargetTrace}}
\longrightarrow
\text{\textsc{TargetReached}}
\\[3pt]
\text{Initial candidate, }|\mathcal R^*(t_0)|=2
&:&
\text{\textsc{Off}}
\longrightarrow
\text{\textsc{ReferenceTrace}}
\longrightarrow
\text{\textsc{TargetReached}}
\\[3pt]
\text{Remaining candidate}
&:&
\text{\textsc{Off}}
\longrightarrow
\text{\textsc{TargetTrace}}
\longrightarrow
\text{\textsc{TargetReached}}
\\[3pt]
\text{Secondary critical robot}
&:&
\text{\textsc{Off}}
\longrightarrow
\lambda^-_{\textsc{Anchor}}
\longrightarrow
\text{\textsc{TargetTrace}}
\longrightarrow
\text{\textsc{TargetReached}}
\\[3pt]
\text{Critical robot}
&:&
\text{\textsc{Off}}
\longrightarrow
\lambda_{\textsc{Anchor}}
\longrightarrow
\text{\textsc{TargetTrace}}
\longrightarrow
\text{\textsc{TargetReached}}.
\end{array}
\]

\noindent\textbf{Case-IV (Configurations
$\mathcal K_{\mathrm{mcr}}^{\mathrm{ii}}$ and
$\mathcal K_{\mathrm{mcr}}^{\mathrm{io}}$).}
In both configurations, at least one non-critical robot lies inside or
on the boundary of $\mathcal D^*(t_0)$. For a configuration in
$\mathcal K_{\mathrm{mcr}}^{\mathrm{ii}}$, its closest unreserved target
point also lies in $\mathcal D^*(t_0)$, whereas for a configuration in
$\mathcal K_{\mathrm{mcr}}^{\mathrm{io}}$, that target point lies
outside $\mathcal D^*(t_0)$. Accordingly, define
$\mathcal P_{\mathrm{adm}}(t_0)
=
\begin{cases}
\mathcal P_{\mathcal N}(t_0)\cap\mathcal D^*(t_0),
&
\text{if }
\mathcal R(t_0)\in\mathcal K_{\mathrm{mcr}}^{\mathrm{ii}},\\[4pt]
\mathcal P_{\mathcal N}(t_0)\setminus\mathcal D^*(t_0),
&
\text{if }
\mathcal R(t_0)\in\mathcal K_{\mathrm{mcr}}^{\mathrm{io}}.
\end{cases}$

\begin{algorithm}[ht]
\caption{$\textsc{ProcMCR-II-IO}(r,\mathcal R(t),\lambda_{\mathbf C}^*,\mathcal P,f_{\mathrm{opt}})$ --- configurations $\mathcal K_{mcr}^{\mathrm{ii}}$, $\mathcal K_{mcr}^{\mathrm{io}}$}
\label{alg:mcriiio}
\begin{algorithmic}[1]
\State $\mathcal P_{\mathrm{adm}}(t)\gets \mathcal P_{\mathcal N}(t)\cap\mathcal D^*(t)$ \textbf{if} $\mathcal K_{mcr}^{\mathrm{ii}}$ \textbf{else} $\mathcal P_{\mathcal N}(t)\setminus\mathcal D^*(t)$
\State $(r_j,p_j)\gets$ candidate--target pair selected from $\mathcal N_{\mathrm{can}}(t)\cap\mathcal D^*(t)$ and $\mathcal P_{\mathrm{can}}(t)\cap\mathcal P_{\mathrm{adm}}(t)$
\State $r_s\gets$ secondary critical robot, if any
\If{$r\in\mathcal R^*(t)$}
    \State $r.\text{light}\gets\lambda_{\textsc{Anchor}}$ (\textbf{if} $\textsc{Off}$); \textbf{remain stationary until all non-critical robots settle}
\ElsIf{$r=r_s$}
    \State $r.\text{light}\gets\lambda^-_{\textsc{Anchor}}$ (\textbf{if} $\textsc{Off}$); \textbf{remain stationary until all candidates reach $\textsc{TargetReached}$}
\ElsIf{$r=r_j$}
    \State $r.\text{light}\gets\textsc{TargetTrace}$; \textbf{move along } $[r(t),p_j]$ \Comment{may cross into/out of $\mathcal D^*(t)$, target unchanged}
\Else
    \State \textbf{remain stationary}
\EndIf
\If{$r$ reaches its assigned target} $r.\text{light}\gets\textsc{TargetReached}$ \EndIf
\end{algorithmic}
\end{algorithm}

Before a candidate moves, every critical robot changes its light from
$\textsc{Off}$ to $\lambda_{\textsc{Anchor}}$ and remains stationary.
The secondary critical robot $r_s$, if it exists, changes its light
from $\textsc{Off}$ to $\lambda^-_{\textsc{Anchor}}$ and also remains
stationary.
Let $(r_j,p_j)$ be selected according to the candidate-selection rule
defined above, with the eligible robot set restricted to
$\mathcal N_{\mathrm{can}}(t_0)\cap\mathcal D^*(t_0)$ and the candidate
target set restricted to
$\mathcal P_{\mathrm{can}}(t_0)\cap\mathcal P_{\mathrm{adm}}(t_0)$. If several pairs attain the
same minimum distance, the tie is resolved using the one-axis
agreement.
The selected candidate $r_j$ changes its light according to
$\text{\textsc{Off}}\longrightarrow
\text{\textsc{TargetTrace}}$ and moves along
$[r_j(t_0),p_j]$. While $r_j$ displays $\textsc{TargetTrace}$, every
other non-critical robot retains $\textsc{Off}$, and the critical and
secondary critical robots retain their respective anchor lights. All
of them remain stationary.
If the movement is interrupted because of non-rigidity, $r_j$ retains
\textsc{TargetTrace} and continues toward $p_j$ in its next
activation. No new candidate is selected until $r_j$ reaches $p_j$.
For $\mathcal K_{\mathrm{mcr}}^{\mathrm{ii}}$, both $r_j(t_0)$ and
$p_j$ lie in $\mathcal D^*(t_0)$. For
$\mathcal K_{\mathrm{mcr}}^{\mathrm{io}}$, the robot initially lies in
$\mathcal D^*(t_0)$ but leaves this region while moving toward
$p_j\notin\mathcal D^*(t_0)$. In either case, it retains the same
target point throughout its movement.
Upon reaching $p_j$, robot $r_j$ changes its light according to
$\text{\textsc{TargetTrace}}\longrightarrow
\text{\textsc{TargetReached}}$ and remains stationary. Its position
fixes the target-point set. The remaining candidate robots then move
toward their assigned target points one at a time using the same light
transition.
After all candidate robots display $\textsc{TargetReached}$, the secondary critical robot changes its light from $\lambda^-_{\textsc{Anchor}}$ to $\textsc{TargetTrace}$ and moves toward its assigned target point. Finally, after every non-critical
robot reaches its target, the critical robots move toward their
assigned critical target points in the deterministic order used in
the procedure for $\mathcal K_{\mathrm{ucr}}^{\mathrm{nc}}$. Every
robot changes its light to $\textsc{TargetReached}$ upon reaching its
target and remains stationary thereafter. The robots follow the common light-transition pattern given in~\ref{app:common-light-transitions}.

\subsection{Correctness of \textsc{MinMaxUniformCircle()}}
\label{subsec:correctness}
The correctness of the algorithm follows from the invariance of the
selected optimal solution and the finite progress of robots toward their
assigned targets. We first prove that the target-point set and the
corresponding assignment remain unchanged during the execution. We then
show that a selected candidate preserves its assigned target until
reaching it, and that simultaneously selected candidates have distinct
reserved targets. These properties ensure collision-free movements.
Finally, the fairness of the scheduler and the non-rigid movement assumption
guarantee that every robot reaches its target in finite time. Hence, the
final configuration realizes the initial optimal assignment and achieves
the optimal min--max travel distance.

\subsubsection{Unique-Critical Invariants}

\begin{lemma}[Invariance of the target-point set for $\mathcal K_{\mathrm{ucr}}^{\mathrm{nc}}$]
\label{lem:unique-critical-invariance}
Let $\mathcal R(t_0)\in\mathcal K_{\mathrm{ucr}}^{\mathrm{nc}}$ have the unique critical robot $r_c$, and let $(\mathcal P_0,f_{\mathrm{opt}})$ be the selected optimal solution with value $\lambda_0$. If a non-critical candidate $r_j$ moves only toward $f_{\mathrm{opt}}(r_j)$ while every other robot remains stationary,
then, throughout its movement, $\lambda_{\mathbf C}^* \bigl(\bar{\mathcal A},\mathcal R(t)\bigr)=\lambda_0,\ \mathcal P(t)=\mathcal P_0.$ Moreover, $f_{\mathrm{opt}}$ remains optimal, $\mathcal R^*(t)=\{r_c\}$, and $r_j$ remains non-critical.
\end{lemma}

\begin{proof}
Section~4.1 of~\cite{bhattacharya2009optimal} associates each robot
and each feasible value $\lambda$ with the arc of the circumference
that the robot can reach within distance $\lambda$. Since $r_c$ is
not located at $\mathcal O$, the point $p_c$ is its unique closest
point on the circumference of $\mathbf C$. Therefore, for every point
$p$ on the circumference, $d\bigl(r_c(t_0),p\bigr)
\geq
d\bigl(r_c(t_0),p_c\bigr)
=
\lambda_0,$ with equality only for $p=p_c$. 

We first establish the upper bound. Since $r_j$ moves directly toward
$p_j$, its distance from $p_j$ never increases. After any positive
movement, $d\bigl(r_j(t),p_j\bigr)
<
d\bigl(r_j(t_0),p_j\bigr)
<
\lambda_0.$ All other non-critical robots remain stationary at distances strictly
smaller than $\lambda_0$ from their retained targets, whereas $d\bigl(r_c(t),p_c\bigr)=\lambda_0.$ Hence, the original target-point set $\mathcal P_0$ and the retained
assignment $f_{\mathrm{opt}}$ remain feasible with maximum distance
$\lambda_0$. Consequently, $\lambda_{\mathbf C}^*
\bigl(\bar{\mathcal A},\mathcal R(t)\bigr)
\leq
\lambda_0.$ For the lower bound, consider any regular target-point set and any
assignment at time $t$. Since every target point lies on the
circumference of $\mathbf C$, the stationary robot $r_c$ must be
assigned to some point $p$ satisfying $d\bigl(r_c(t),p\bigr)\geq\lambda_0.$ Thus, no value smaller than $\lambda_0$ is feasible, and hence $\lambda_{\mathbf C}^*
\bigl(\bar{\mathcal A},\mathcal R(t)\bigr)
\geq
\lambda_0.$ Combining the two bounds gives $\lambda_{\mathbf C}^*
\bigl(\bar{\mathcal A},\mathcal R(t)\bigr)
=
\lambda_0.$ Furthermore, every optimal assignment at value $\lambda_0$ must assign
$r_c$ to $p_c$, because $p_c$ is the unique circumference point at
distance $\lambda_0$ from $r_c$. Therefore, every optimal regular
target-point set contains $p_c$. Once $p_c$ is fixed, the remaining
$n-1$ target points are uniquely determined at angular separations
$2\pi/n$. Hence, the target-point set remains
$\mathcal P(t)=\mathcal P_0$. The light \textsc{\(\lambda\)-Anchor} keeps the stationary critical
robot $r_c$ identifiable throughout the candidate movement. Thus,
after every non-rigid interruption, robot $r_j$ reconstructs the same
point $p_c$, the same target-point set $\mathcal P_0$, and continues
toward the same target $p_j$ with light \textsc{TargetTrace}. When $r_j$ reaches $p_j$ and changes its light to
\textsc{TargetReached}, its position provides an additional reference
for $\mathcal P_0$. The remaining target points can then be computed
at angular separations $2\pi/n$, and the next eligible candidate moves
with respect to the same invariant target-point set.
\end{proof}

Let $\gamma_c
=
\lambda_0-
\max_{r_i\in\mathcal R(t_0)\setminus\{r_c\}}
d\bigl(r_i(t_0),f_{\mathrm{opt}}(r_i)\bigr).$ Since $r_c$ is the unique critical robot, $\gamma_c>0$.

\begin{lemma}[centre-critical reduction]
\label{lem:centre-invariant-reduction}
Let $q_c$ be the point on $\mathbf L$ satisfying
$d(\mathcal O,q_c)=\delta$, where $0<\delta<\min\{\rho_c,\gamma_c\}.$ Suppose that $r_c$ moves from $\mathcal O$ toward $q_c$ along
$\mathbf L$, while every other robot remains stationary. Then, at
every intermediate time $t$ satisfying $0<\sigma_t
=
d\bigl(\mathcal O,r_c(t)\bigr)
\leq\delta,$ the target-point set remains $\mathcal P_0$, $r_c$ remains the unique
critical robot, and $\mathcal R(t)\in\mathcal K_{ucr}^{nc}.$
\end{lemma}

\begin{proof}
Since $r_c$ moves only along $\mathbf L$, the ray starting from
$\mathcal O$ and passing through $r_c(t)$ remains $\mathbf L$.
Therefore, its intersection with the circumference of $\mathbf C$
remains $p_c$. Since $p_c$ and the angular separation $2\pi/n$
uniquely determine the regular target-point set, we have $\mathcal P(t)=\mathcal P_0.$ Let $h$ be the radius of $\mathbf C$. Since
$r_c(t_0)=\mathcal O$ and $r_c$ is critical, $\lambda_0=h$. At time
$t$, $d\bigl(r_c(t),p_c\bigr)
=
h-\sigma_t
=
\lambda_0-\sigma_t.$ For every $r_i\neq r_c$, the robot remains stationary and hence $d\bigl(r_i(t),f_{\mathrm{opt}}(r_i)\bigr)
\leq
\lambda_0-\gamma_c
<
\lambda_0-\delta
\leq
\lambda_0-\sigma_t.$ Thus, under the retained assignment, $r_c$ is the unique robot
attaining the maximum distance. Moreover, every point of $\mathbf C$ is at distance at least $h-\sigma_t$ from $r_c(t)$, with equality only at $p_c$. Therefore,
no assignment can attain a min--max value smaller than
$h-\sigma_t$, while the retained assignment attains this value.
Hence, it remains optimal and $r_c$ remains the unique critical robot. Finally, $\delta<\rho_c$ ensures that the segment
$[\mathcal O,q_c]$ contains no other robot. Since
$r_c(t)\neq\mathcal O$ for every $\sigma_t>0$, we obtain $\mathcal R(t)\in\mathcal K_{ucr}^{nc}.$
\end{proof}

\begin{corollary}[Preservation of the original min--max travel bound]
\label{cor:centre-travel-bound}
Let $s=d(q_c,p_c).$ Since $\mathcal O$, $q_c$, and $p_c$ are collinear and $q_c$ lies
between $\mathcal O$ and $p_c$, $\delta+s
=
d(\mathcal O,q_c)+d(q_c,p_c)
=
d(\mathcal O,p_c)
=
h
=
\lambda_0.$ Therefore, the movement through $q_c$ only divides the original
optimal path of $r_c$ into two consecutive segments and does not
violate the original min--max travel bound.
\end{corollary}

By Lemma~\ref{lem:centre-invariant-reduction}, after the initial
movement of $r_c$, the resulting configuration belongs to
$\mathcal K_{ucr}^{nc}$ with the same target-point set. Therefore, the
remaining execution follows from the correctness of the procedure for
$\mathcal K_{ucr}^{nc}$. Moreover, by
Corollary~\ref{cor:centre-travel-bound}, the original optimal min--max travel bound is preserved.

\subsubsection{Multiple-Critical Invariants}

\begin{lemma}[Invariance of the target-point set for multiple-critical configurations]
\label{lem:mcr-target-invariance}
Let $\mathcal R(t_0)\in\mathcal K_{\mathrm{mcr}}$, and let
$\mathcal P_0$ be the target-point set selected at time $t_0$. If the
robots execute the corresponding procedure described in Cases I--IV,
then, at every time $t$ during the execution, $\mathcal P(t)=\mathcal P_0.$
\end{lemma}

\begin{proof}
Let $t_1$ be the first time at which a robot displays
$\textsc{TargetReached}$. We first prove that
$\mathcal P(t)=\mathcal P_0$ for every $t<t_1$.

\medskip
\noindent
\textit{Case-I:
$\mathcal K_{\mathrm{mcr}}^{\mathrm{all}}$.}
The unselected critical robots retain
$\lambda_{\textsc{Anchor}}$ and remain stationary while the selected
critical robots move. Therefore, the reference positions and lights
used to determine the initial orientation remain unchanged. Every
interrupted robot displaying $\textsc{TargetTrace}$ consequently
reconstructs the same target-point set $\mathcal P_0$. This requires
that the selected robots do not exhaust the complete reference set
before one of them reaches its target.

\medskip
\noindent
\textit{Case-II:
$\mathcal K_{\mathrm{mcr}}^{\mathrm{ot}}$.}
Every critical robot retains $\lambda_{\textsc{Anchor}}$, and the
secondary critical robot, if one exists, retains
$\lambda^-_{\textsc{Anchor}}$. All these robots remain stationary
while the selected candidate moves. Hence, the reference information
used to reconstruct the target-point set remains unchanged, and every
intermediate snapshot yields $\mathcal P_0$.

\medskip
\noindent
\textit{Case-III:
$\mathcal K_{\mathrm{mcr}}^{\mathrm{ct}}$.}
If $|\mathcal R^*(t_0)|>2$, the stationary critical robots displaying
$\lambda_{\textsc{Anchor}}$ preserve the prescribed reference
information. Therefore, each candidate displaying
$\textsc{TargetTrace}$ reconstructs $\mathcal P_0$ after every
interrupted movement. Now suppose that $|\mathcal R^*(t_0)|=2$. The two critical robots
$r_a^*$ and $r_b^*$ remain stationary. Moreover, the selected set $S$
is fixed before any of its members moves and is excluded from the
count defining $H^+$. Consequently, the midpoint $m_{ab}$, the region
$H^+$, the point $z^+$, and the dummy point $w=\frac{m_{ab}+z^+}{2}$ remain unchanged. Hence, every robot displaying $\textsc{ReferenceTrace}$ applies the same dummy-point reconstruction rule and obtains the same target-point set $\mathcal P_0$.

\medskip
\noindent
\textit{Case-IV:
$\mathcal K_{\mathrm{mcr}}^{\mathrm{ii}}$ or
$\mathcal K_{\mathrm{mcr}}^{\mathrm{io}}$.}
The critical and secondary critical robots retain their respective
anchor lights and remain stationary while the selected candidate
moves. Therefore, the reference information remains unchanged,
irrespective of whether the candidate stays inside
$\mathcal D^*(t_0)$ or leaves this region. Thus, the candidate
reconstructs the same target-point set $\mathcal P_0$ throughout its
movement. It remains to consider the execution at and after $t_1$. Let $p_j$ be
the target occupied by the first robot displaying
$\textsc{TargetReached}$. Since $p_j\in\mathcal P_0$, the point $p_j$
together with the angular separation $2\pi/n$ uniquely determines the
regular target-point set $\mathcal P_0$, up to cyclic relabelling.
The robot at $p_j$ retains $\textsc{TargetReached}$ and remains
stationary. Hence, from time $t_1$ onward, every robot reconstructs
the target-point set directly from $p_j$, without using the earlier
critical or dummy reference. Therefore, $\mathcal P(t)=\mathcal P_0$ throughout the complete execution.
\end{proof}

\subsubsection{Supporting Invariants for Multiple-Critical Configurations}
\label{app:mcr-supporting-invariants}

\begin{lemma}[Preservation of the retained optimal solution]
\label{lem:mcr-optimal-solution-invariance}
Let
$\mathcal R(t_0)\in
\mathcal K_{\mathrm{mcr}}\setminus
\mathcal K_{\mathrm{mcr}}^{\mathrm{all}}$, and let
$(\mathcal P_0,f_{\mathrm{opt}})$ be the selected optimal solution with
value
$\lambda_0=
\lambda_{\mathbf C}^*
(\bar{\mathcal A},\mathcal R(t_0))$.
Let $t_c$ be the first time at which an initially critical robot starts
moving. Then, for every $t\in[t_0,t_c)$, $\lambda_{\mathbf C}^*
\bigl(\bar{\mathcal A},\mathcal R(t)\bigr)
=
\lambda_0,$ and the retained assignment $f_{\mathrm{opt}}$ remains optimal with
respect to the invariant target-point set $\mathcal P_0$. In
particular, $\max_{r_i\in\mathcal R(t)}
d\bigl(r_i(t),f_{\mathrm{opt}}(r_i)\bigr)
=
\lambda_0.$
\end{lemma}

\begin{proof}
Fix any $t\in[t_0,t_c)$. By
Lemma~\ref{lem:mcr-target-invariance}, the target-point set remains
$\mathcal P_0$ throughout this interval. From the critical characterization used in the computation of $\lambda_0$, for every admissible regular target-point set
$\mathcal P(\theta)$ and every injective assignment
$g:\mathcal R^*(t_0)\rightarrow\mathcal P(\theta)$, we have
\[
\max_{r_i^*\in\mathcal R^*(t_0)}
d\bigl(r_i^*(t_0),g(r_i^*)\bigr)
\geq
\lambda_0.
\tag{C}
\]

Since $t<t_c$, every initially critical robot remains stationary.
Therefore, for each $r_i^*\in\mathcal R^*(t_0)$, $d\bigl(r_i^*(t),f_{\mathrm{opt}}(r_i^*)\bigr)
=
d\bigl(r_i^*(t_0),f_{\mathrm{opt}}(r_i^*)\bigr)
=
\lambda_0.$ Every initially non-critical robot has assigned distance strictly
smaller than $\lambda_0$. By the movement procedure, such a robot
either remains stationary or moves only toward its retained target.
Hence, for every $r_j\in\mathcal N(t_0)$, $d\bigl(r_j(t),f_{\mathrm{opt}}(r_j)\bigr)
\leq
d\bigl(r_j(t_0),f_{\mathrm{opt}}(r_j)\bigr)
<
\lambda_0.$ Consequently, $\max_{r_i\in\mathcal R(t)}
d\bigl(r_i(t),f_{\mathrm{opt}}(r_i)\bigr)
=
\lambda_0.$ Thus, the retained solution
$(\mathcal P_0,f_{\mathrm{opt}})$ remains feasible at value
$\lambda_0$, and therefore $\lambda_{\mathbf C}^*
\bigl(\bar{\mathcal A},\mathcal R(t)\bigr)
\leq
\lambda_0.$ Conversely, consider any admissible target-point set
$\mathcal P(\theta)$ and any assignment $f$ for $\mathcal R(t)$.
Its restriction to $\mathcal R^*(t_0)$ is an injective assignment.
Since the initially critical robots remain stationary,
Condition~$(\mathrm C)$ gives
\[
\begin{aligned}
\max_{r_i\in\mathcal R(t)}
d\bigl(r_i(t),f(r_i)\bigr)
&\geq
\max_{r_i^*\in\mathcal R^*(t_0)}
d\bigl(r_i^*(t),f(r_i^*)\bigr)\\
&\geq
\lambda_0.
\end{aligned}
\]
Hence, no solution having value smaller than $\lambda_0$ exists, and $\lambda_{\mathbf C}^*
\bigl(\bar{\mathcal A},\mathcal R(t)\bigr)
\geq
\lambda_0.$ Combining the two inequalities yields $\lambda_{\mathbf C}^*
\bigl(\bar{\mathcal A},\mathcal R(t)\bigr)
=
\lambda_0.$ Since $(\mathcal P_0,f_{\mathrm{opt}})$ attains this value,
$f_{\mathrm{opt}}$ remains optimal throughout $[t_0,t_c)$.
\end{proof}

\begin{lemma}[Persistence of the critical set]
\label{lem:critical-set-persistence}
Let $\mathcal R(t_0)\notin
\mathcal K_{\mathrm{mcr}}^{\mathrm{all}}$, and let $t_c$ be the first
time at which an initially critical robot starts moving. Then, for
every $t\in[t_0,t_c)$, $\mathcal R^*(t)=\mathcal R^*(t_0).$
\end{lemma}

\begin{proof}
Let
$\lambda_0=\lambda_{\mathbf C}^*
(\bar{\mathcal A},\mathcal R(t_0))$.
By the applicable result among
Lemmas~\ref{lem:unique-critical-invariance}--%
\ref{lem:mcr-target-invariance}, the target-point set, the retained
assignment, and the optimal min--max value remain invariant before any
critical robot moves. Every $r_i^*\in\mathcal R^*(t_0)$ remains stationary during this
interval. Hence, for its reserved target $p_i^*$, $d\bigl(r_i^*(t),p_i^*\bigr)
=
d\bigl(r_i^*(t_0),p_i^*\bigr)
=
\lambda_0.$ Thus, every initially critical robot remains critical. For every initially non-critical robot $r_j$, its assigned distance is
strictly smaller than $\lambda_0$. If $r_j$ remains stationary, this
distance does not change; if it moves, it moves toward its fixed target
point, and hence the distance decreases. Therefore, $r_j$ cannot
become critical. By Lemmas~\ref{lem:mcr-target-invariance} and
\ref{lem:mcr-optimal-solution-invariance}, the target-point set, the
retained optimal assignment, and the optimal min--max value remain
unchanged until an initially critical robot starts moving. Consequently,
$\mathcal R^*(t)=\mathcal R^*(t_0)$ for every $t<t_c$.
\end{proof}

\subsubsection{Persistence of Candidates and Target Reservations}

\begin{lemma}[Invariance of the assigned target point]
\label{lem:assigned-target-invariance}
Let $(r_j,p_j)$ be the candidate--target pair selected at time $t_s$.
By the applicable result among
Lemmas~\ref{lem:unique-critical-invariance}--%
\ref{lem:mcr-target-invariance}, the target-point set remains invariant.
Under the movement procedure, $r_j$ moves only along
$[r_j(t_s),p_j]$ and retains its movement light. Then $p_j$ remains
assigned to $r_j$ until $r_j$ reaches it.
\end{lemma}

\begin{proof}
Let $\mathcal P_{\mathrm{av}}$ be the set of target points available to
$r_j$ at time $t_s$. By the candidate-selection rule,
$d(r_j(t_s),p_j)\leq d(r_j(t_s),q)$ for every
$q\in\mathcal P_{\mathrm{av}}$.

At any intermediate time $t$, there exists $\alpha\in[0,1]$ such that
$r_j(t)=(1-\alpha)r_j(t_s)+\alpha p_j$. For every
$q\in\mathcal P_{\mathrm{av}}\setminus\{p_j\}$, we have
\[
\begin{aligned}
&d^2\bigl(r_j(t),q\bigr)-d^2\bigl(r_j(t),p_j\bigr)\\
&\quad=(1-\alpha)
\left[
d^2\bigl(r_j(t_s),q\bigr)
-
d^2\bigl(r_j(t_s),p_j\bigr)
\right]
+\alpha d^2(p_j,q).
\end{aligned}
\]
The first term is non-negative, and the second term is strictly
positive whenever $\alpha>0$. Hence, after any positive movement,
$p_j$ is the unique closest available target point of $r_j$. At
$\alpha=0$, any initial tie is resolved by the prescribed one-axis
agreement. If the movement is interrupted, $r_j$ retains
$\textsc{TargetTrace}$, or $\textsc{ReferenceTrace}$ in the
corresponding case, and every other robot remains stationary.
Therefore, the available target-point set does not change, and every
subsequent activation of $r_j$ yields the same target point $p_j$.
Consequently, the assignment of $r_j$ remains unchanged until it
reaches $p_j$.
\end{proof}

\begin{lemma}[Persistence and exclusive reservation of candidate targets]
\label{lem:candidate-persistence-reservation}
Let $(r_j,p_j)$ be a candidate--target pair selected at time $t_s$.
Then $r_j$ retains $p_j$ as its assigned target until reaching it.
Moreover, distinct simultaneously selected candidates are assigned
distinct target points, and every such target remains reserved for its
candidate until it is reached.
\end{lemma}

\begin{proof}
Let $(r_j,p_j)$ be a candidate--target pair selected at time $t_s$.
By the applicable target-set invariance result, the target-point set
$\mathcal P$ remains unchanged throughout the movement of $r_j$.
Moreover, by Lemma~\ref{lem:assigned-target-invariance}, once $r_j$
starts moving toward $p_j$, the point $p_j$ remains its assigned
target until it is reached. In particular, if the movement is
interrupted because of non-rigidity, $r_j$ retains
$\textsc{TargetTrace}$, or $\textsc{ReferenceTrace}$ in the
corresponding case, and resumes its movement toward the same point
$p_j$ at its next activation. Hence, the candidate--target pair
$(r_j,p_j)$ persists until $r_j$ reaches $p_j$. Now consider two distinct candidates $r_i,r_j$ selected from the same
snapshot. Since the selected optimal assignment $f_{\mathrm{opt}}$ is
a bijection, $f_{\mathrm{opt}}(r_i)\neq f_{\mathrm{opt}}(r_j).$ Thus, simultaneously selected candidates are assigned distinct target
points. Furthermore, once $p_j$ is assigned to $r_j$, it remains reserved
while $r_j$ displays $\textsc{TargetTrace}$ or
$\textsc{ReferenceTrace}$ and is therefore excluded from every
subsequent candidate selection. When $r_j$ reaches $p_j$, it changes
its light to $\textsc{TargetReached}$ and remains stationary at
$p_j$; hence, $p_j$ remains unavailable to all subsequently selected
robots. Consequently, every selected candidate retains its assigned
target until reaching it, and no two candidates can ever reserve the
same target point.
\end{proof}

\subsubsection{Collision Freedom and Finite Progress}

\begin{theorem}[Collision-Free Execution of
$\textsc{MinMaxUniformCircle()}$]
\label{thm:collision-freedom}
If all robots initially occupy pairwise distinct positions, then every
execution of $\textsc{MinMaxUniformCircle()}$ preserves this property,
provided that every authorized movement phase obeys the clear-path rule.
\end{theorem}

\begin{proof}
By Lemma~\ref{lem:candidate-persistence-reservation}, the targets assigned to
the robots in $S$ are pairwise distinct. By the clear-path rule, their
movement paths satisfy the required disjointedness conditions. Consider an arbitrary movement phase starting at time $t_s$. Let $S$ denote the set of robots authorized to move during this phase,
and let $p_i$ be the target assigned to each robot $r_i\in S$.
By the clear-path rule, the following conditions hold:

\begin{enumerate}
    \item the assigned targets $\{p_i : r_i\in S\}$ are pairwise distinct
    and unoccupied at time $t_s$;

    \item for every $r_i\in S$, $(r_i(t_s),p_i]
    \cap
    \{r_k(t_s):r_k\notin S\}
    =
    \emptyset;$

    \item for any two distinct robots $r_i,r_j\in S$, $[r_i(t_s),p_i]
    \cap
    [r_j(t_s),p_j]
    =
    \emptyset.$
\end{enumerate}

By Lemma~\ref{lem:assigned-target-invariance}, every robot in $S$
preserves its assigned target throughout the phase. Consequently, for
every time $t$ during this phase, $r_i(t)\in [r_i(t_s),p_i],
\ r_i\in S.$ Under non-rigid motion, an adversarial interruption may stop a robot at an intermediate point of this segment, but cannot alter its prescribed
trajectory. Moreover, while robots display $\textsc{TargetTrace}$ or
$\textsc{ReferenceTrace}$, no robot outside $S$ is allowed to initiate
a movement. Robots displaying $\lambda_{\textsc{Anchor}}$ or
$\lambda^-_{\textsc{Anchor}}$ also remain stationary. Hence, during
the considered phase, only robots in $S$ may change their positions. Assume, for contradiction, that a collision occurs during this phase,
and consider the first such collision. There are only two possibilities. If a moving robot $r_i\in S$ collides with a stationary robot
$r_k\notin S$, then $r_k(t_s)\in (r_i(t_s),p_i],$ contradicting Condition~2. Otherwise, two robots $r_i,r_j\in S$ collide. Their common collision
point must then belong to both prescribed movement segments, implying $[r_i(t_s),p_i]
\cap
[r_j(t_s),p_j]
\neq
\emptyset,$ which contradicts Condition~3. Thus, no collision can occur during the phase. In particular, when a
robot $r_i$ reaches $p_i$, it changes its light to
$\textsc{TargetReached}$ and remains stationary there; since $p_i$ is
included in its movement segment, Condition~3 already prevents any
simultaneously moving robot from reaching the same position. In every
subsequent movement phase, this stationary robot is again taken into
account by the clear-path rule. The preliminary displacement of a centre-critical robot is also collision-free, since it is performed entirely inside an empty sector
and within its prescribed safe distance $\rho_c$. Therefore, every movement phase of $\textsc{MinMaxUniformCircle()}$ is collision-free. Since the robots
occupy pairwise distinct positions initially, induction over successive
movement phases implies that no two robots ever occupy the same
position. Hence, every execution of $\textsc{MinMaxUniformCircle()}$ is collision-free.
\end{proof}

\begin{lemma}[Finite-time achievement]
\label{lem:finite-time-achievement}
Under a fair ASYNC scheduler and the non-rigid movement model with
minimum progress $\delta>0$, every execution of
$\textsc{MinMaxUniformCircle()}$ terminates after finitely many
Look--Compute--Move cycles. At termination, every robot occupies a
distinct target point of $\mathcal P$, and hence
$\mathcal{MMUCF}$ is achieved.
\end{lemma}

\begin{proof}
Let $r_i$ be a robot selected at time $t_s$, and let $p_i$ be its
assigned target point. By Lemmas~\ref{lem:assigned-target-invariance}
and~\ref{lem:candidate-persistence-reservation}, the pair $(r_i,p_i)$ remains
unchanged until $r_i$ reaches $p_i$. Let $D_i=d(r_i(t_s),p_i)$. Under the non-rigid movement model, every
effective Move of $r_i$ either brings it directly to $p_i$ or decreases
its remaining distance by at least $\delta$. Therefore, $r_i$ reaches
$p_i$ after at most $\left\lceil\frac{D_i}{\delta}\right\rceil$ Effective move operations. Fairness guarantees that every selected robot is activated sufficiently often to complete these movements. If several candidates are selected from the same snapshot, the same argument applies independently to each of them because their target points are distinct and remain reserved by Lemma~\ref{lem:candidate-persistence-reservation}. Now consider any nonterminal configuration. If a robot displays
$\textsc{TargetTrace}$ or $\textsc{ReferenceTrace}$, it remains
eligible and continues toward its fixed target. Otherwise, the
selection rule of the applicable configuration class selects at least
one eligible robot--target pair. After the candidate robots have
settled, the secondary critical robot and subsequently the critical
robots become eligible according to the movement procedure. Hence, no
nonterminal configuration is a deadlock. Whenever a selected robot reaches its target, it changes its light to
$\textsc{TargetReached}$ and never moves again. Thus, every completed
target-reaching phase increases the number of robots displaying
$\textsc{TargetReached}$ by at least one. Since there are only $n$
robots, at most $n$ such phases can occur. The preliminary displacement of the centre-critical robot is the only
phase that may not immediately place a robot at a target point. This
movement has finite length and, by the same minimum-progress argument,
is completed after finitely many effective Moves. It occurs only once
and transforms the configuration into
$\mathcal K_{\mathrm{ucr}}^{\mathrm{nc}}$, after which the usual
target-reaching procedure is followed. Consequently, every selected movement is completed after finitely many
effective Moves, and only finitely many movement phases occur. Under
the fair ASYNC scheduler, the execution therefore terminates after
finitely many Look--Compute--Move cycles. By Lemmas~\ref{lem:unique-critical-invariance},
\ref{lem:centre-invariant-reduction}, and
\ref{lem:mcr-target-invariance}, all robots use the same invariant
regular target-point set $\mathcal P$. Moreover, the selected optimal
assignment is bijective, and assigned targets remain distinct and
reserved. Hence, at termination, the $n$ robots occupy the $n$ distinct
points of $\mathcal P$, one robot at each target point. Therefore,
$\mathcal{MMUCF}$ is achieved.
\end{proof}

\begin{theorem}[Correctness of
$\textsc{MinMaxUniformCircle()}$]
\label{thm:minmax-uniform-circle-correctness}
Under a fair $\mathcal{ASYNC}$ scheduler and non-rigid movements,
$\textsc{MinMaxUniformCircle()}$ terminates in finite time without
collisions, places the robots at the distinct points of
$\mathcal P_0$, and solves \textsc{MMUCF} optimally with maximum
travel distance
$\lambda_{\mathbf C}^*(\bar{\mathcal A},\mathcal R(t_0))$.
\end{theorem}

\begin{proof}
The classification based on the critical set covers all possible
configurations. For $\mathcal K_{\mathrm{ucr}}^{\mathrm{nc}}$ and $\mathcal K_{\mathrm{ucr}}^{\mathrm{c}}$, Lemmas~\ref{lem:unique-critical-invariance} and
\ref{lem:centre-invariant-reduction} preserve the optimal solution and
reduce the execution to the non-centre case. For multiple-critical
configurations, Lemma~\ref{lem:mcr-target-invariance} preserves the
target-point set throughout the execution. By Lemmas~\ref{lem:mcr-optimal-solution-invariance}--\ref{lem:candidate-persistence-reservation},
the selected optimal assignment remains invariant, and every selected
candidate preserves its assigned target until reaching it. Hence,
asynchronous activations and non-rigid interruptions do not affect the
execution. By Theorem~\ref{thm:collision-freedom} and
Lemma~\ref{lem:finite-time-achievement}, the execution is collision-free
and terminates in finite time. Since $f_{\mathrm{opt}}$ is a bijection,
the robots finally occupy all points of $\mathcal P_0$. It remains to prove optimality. For each robot $r_i$, let $L_i$ denote
the total distance travelled by $r_i$ during the execution. Since every
robot moves only along the segment joining its initial position to its
assigned target, and non-rigid movements only split this segment into
smaller parts, we have $L_i=d(r_i(t_0),f_{\mathrm{opt}}(r_i)).$ The preliminary movement in the center-critical case is collinear with
the final movement and, by Corollary~\ref{cor:centre-travel-bound}, does
not increase the optimal bound. Thus, $\max_i L_i=
\lambda_{\mathbf C}^*
(\bar{\mathcal A},\mathcal R(t_0)).$ No algorithm can achieve a smaller value by the definition of the min--max objective. Hence, $\textsc{MinMaxUniformCircle()}$ is collision-free, terminates in finite time, and solves \textsc{MMUCF} optimally.
\end{proof}





\section{Conclusion}

In this paper, we studied the min--max uniform circle formation problem
for autonomous mobile robots under the asynchronous and non-rigid
movement model. We established three impossibility results and proposed
a distributed algorithm for robots equipped with externally visible
lights and one-axis agreement. The algorithm preserves the optimal
target-point set and assignment despite asynchronous activations and
non-rigid interruptions, guarantees collision-free execution, and
achieves the optimal min--max travel distance in finite time. Several directions remain open, including weakening the robot
capabilities, considering the case where the target circle is not given
in advance, studying limited visibility settings, and developing
fault-tolerant variants in the presence of obstacles.

\bibliographystyle{plain} 
\bibliography{Biblo}

@String{Computing = "Computing" }

@String{Computer = "{IEEE} Computer" }

@book{DBLP:series/synthesis/2012Flocchini,
  author       = {Paola Flocchini and
                  Giuseppe Prencipe and
                  Nicola Santoro},
  title        = {Distributed Computing by Oblivious Mobile Robots},
  series       = {Synthesis Lectures on Distributed Computing Theory},
  publisher    = {Morgan {\&} Claypool Publishers},
  year         = {2012},
  url          = {https://doi.org/10.2200/S00440ED1V01Y201208DCT010},
  doi          = {10.2200/S00440ED1V01Y201208DCT010},
  isbn         = {978-3-031-00880-1},
  bibsource    = {dblp computer science bibliography, https://dblp.org}
}

@article{defago2008non,
  title={Non-uniform circle formation algorithm for oblivious mobile robots with convergence toward uniformity},
  author={D{\'e}fago, Xavier and Souissi, Samia},
  journal={Theoretical Computer Science},
  volume={396},
  number={1-3},
  pages={97--112},
  year={2008},
  publisher={Elsevier}
}

@inproceedings{flocchini1999hard,
  title={Hard tasks for weak robots: The role of common knowledge in pattern formation by autonomous mobile robots},
  author={Flocchini, Paola and Prencipe, Giuseppe and Santoro, Nicola and Widmayer, Peter},
  booktitle={International Symposium on Algorithms and Computation},
  pages={93--102},
  year={1999},
  organization={Springer}
}

@String{Springer = "Springer-Verlag" }

@article{chun2016robot,
  title={Robot surveillance and security},
  author={Chun, Wendell H and Papanikolopoulos, Nikolaos},
  journal={springer Handbook of Robotics},
  pages={1605--1626},
  year={2016},
  publisher={Springer}
}

@inproceedings{flocchini2013rendezvous,
  title={Rendezvous of two robots with constant memory},
  author={Flocchini, Paola and Santoro, Nicola and Viglietta, Giovanni and Yamashita, Masafumi},
  booktitle={International Colloquium on Structural Information and Communication Complexity},
  pages={189--200},
  year={2013},
  organization={Springer}
}

@inproceedings{Defago2002,
 author = {D{\'e}fago, X. and Konagaya, A.},
 title = {Circle formation for oblivious anonymous mobile robots with no common sense of orientation},
 booktitle = {Proceedings of the second ACM international workshop on Principles of mobile computing},
 series = {POMC '02},
 year = {2002},
 isbn = {1-58113-511-4},
 location = {Toulouse, France},
 pages = {97--104},
 numpages = {8},
 url = {http://doi.acm.org/10.1145/584490.584509},
 doi = {10.1145/584490.584509},
 acmid = {584509},
 publisher = {ACM},
 address = {New York, NY, USA},
}

@ARTICLE{Suzuki1999,
    author = {Ichiro Suzuki  and Masafumi Yamashita.},
    title = {Distributed anonymous mobile robots: Formation of geometric patterns},
    journal = {SIAM Journal on Computing},
    year = {1999},
    volume = {28},
    pages = {pages 1347-1363}
}

@INPROCEEDINGS{vigl2016,
    author = {Marcello Mamino and Giovanni Viglietta},
    title = {Square Formation by Asynchronous Oblivious Robots},
    booktitle = {In Proceedings of the 28th Canadian Conference on Computational Geometry (CCCG)},
    year = {2016},
    pages = {1-6},
    publisher = {}
}

@inproceedings{feletti2018uniform,
  title={Uniform circle formation for swarms of opaque robots with lights},
  author={Feletti, Caterina and Mereghetti, Carlo and Palano, Beatrice},
  booktitle={International Symposium on Stabilizing, Safety, and Security of Distributed Systems},
  pages={317--332},
  year={2018},
  organization={Springer}
}

@inproceedings{pattanayak2024time,
  title={Time-color tradeoff on uniform circle formation by asynchronous robots},
  author={Pattanayak, Debasish and Sharma, Gokarna},
  booktitle={2024 IEEE International Parallel and Distributed Processing Symposium (IPDPS)},
  pages={987--997},
  year={2024},
  organization={IEEE}
}

@article{flocchini2017distributed,
  title={Distributed computing by mobile robots: uniform circle formation},
  author={Flocchini, Paola and Prencipe, Giuseppe and Santoro, Nicola and Viglietta, Giovanni},
  journal={Distributed Computing},
  volume={30},
  number={6},
  pages={413--457},
  year={2017},
  publisher={Springer}
}

@article{brambilla2013swarm,
  title={Swarm robotics: a review from the swarm engineering perspective},
  author={Brambilla, Manuele and Ferrante, Eliseo and Birattari, Mauro and Dorigo, Marco},
  journal={Swarm Intelligence},
  volume={7},
  number={1},
  pages={1--41},
  year={2013}
}

@article{bayindir2016review,
  title={A review of swarm robotics tasks},
  author={Bayindir, Levent},
  journal={Neurocomputing},
  volume={172},
  pages={292--321},
  year={2016}
}

@article{corke2005networked,
  title={Networked robots: Flying robot navigation using a sensor net},
  author={Corke, Peter and Peterson, Ron and Rus, Daniela},
  journal={Robotics Research},
  year={2005}
}

@article{schwager2009decentralized,
  title={Decentralized, adaptive coverage control for networked robots},
  author={Schwager, Mac and Julian, Brian and Rus, Daniela},
  journal={The International Journal of Robotics Research},
  volume={28},
  number={3},
  pages={357--375},
  year={2009}
}

@incollection{bhagat2018optimum,
  title={Optimum circle formation by autonomous robots},
  author={Bhagat, Subhash and Mukhopadhyaya, Krishnendu},
  booktitle={Advanced Computing and Systems for Security: Volume Five},
  pages={153--165},
  year={2018},
  publisher={Springer}
}

@article{chen2015optimal,
  title={Optimal point movement for covering circular regions},
  author={Chen, Danny Z and Tan, Xuehou and Wang, Haitao and Wu, Gangshan},
  journal={Algorithmica},
  volume={72},
  number={2},
  pages={379--399},
  year={2015},
  publisher={Springer}
}

@article{bhattacharya2009optimal,
  title={Optimal movement of mobile sensors for barrier coverage of a planar region},
  author={Bhattacharya, Binay and Burmester, Mike and Hu, Yuzhuang and Kranakis, Evangelos and Shi, Qiaosheng and Wiese, Andreas},
  journal={Theoretical Computer Science},
  volume={410},
  number={52},
  pages={5515--5528},
  year={2009},
  publisher={Elsevier}
}

@article{parunak1996applications,
	title={Applications of distributed artificial intelligence in industry},
	author={Parunak, H Van Dyke},
	journal={Foundations of distributed artificial intelligence},
	volume={2},
	year={1996}
}

@inproceedings{steinberg2006intelligent,
	title={Intelligent autonomy for unmanned naval systems},
	author={Steinberg, M.},
	booktitle={ Proc. Defense and Security Symposium},
	pages={623013--623013},
	year={2006},
	organization={International Society for Optics and Photonics}
}

@article{Sruti2016,
  author       = {Sruti Gan Chaudhuri and
                  Krishnendu Mukhopadhyaya},
  title        = {Leader election and gathering for asynchronous fat robots without
                  common chirality},
  journal      = {J. Discrete Algorithms},
  volume       = {33},
  pages        = {171--192},
  year         = {2015},
  url          = { },
  doi          = { },
  bibsource    = { }
}

@article{terai2018gathering,
  title={Gathering problems for autonomous mobile robots with lights},
  author={Terai, Satoshi and Wada, Koichi and Katayama, Yoshiaki},
  journal={arXiv preprint arXiv:1811.12068},
  year={2018}
}

@inproceedings{mondal2019uniform,
  title={Uniform circle formation by swarm robots under limited visibility},
  author={Mondal, Moumita and Gan Chaudhuri, Sruti},
  booktitle={International Conference on Distributed Computing and Internet Technology},
  pages={420--428},
  year={2019},
  organization={Springer}
}

@article{feletti2024optimal,
  title={Optimal uniform circle formation by asynchronous luminous robots},
  author={Feletti, Caterina and Pattanayak, Debasish and Sharma, Gokarna},
  journal={arXiv preprint arXiv:2405.06617},
  year={2024}
}

@inproceedings{di2014robots,
  title={Robots with lights: Overcoming obstructed visibility without colliding},
  author={Di Luna, Giuseppe Antonio and Flocchini, Paola and Gan Chaudhuri, Sruti and Santoro, Nicola and Viglietta, Giovanni},
  booktitle={Symposium on Self-Stabilizing Systems},
  pages={150--164},
  year={2014},
  organization={Springer}
}

@inproceedings{kegeleirs2019random,
  title={Random walk exploration for swarm mapping},
  author={Kegeleirs, Miquel and Garz{\'o}n Ramos, David and Birattari, Mauro},
  booktitle={Annual conference towards autonomous robotic systems},
  pages={211--222},
  year={2019},
  organization={Springer}
}

@article{albiero2021swarm,
  title={Swarm Robots in Agriculture},
  author={Albiero, Daniel and Garcia, Angel Pontin and Umezu, Claudio Kiyoshi and de Paulo, Rodrigo Leme},
  journal={arXiv preprint arXiv:2103.06732},
  year={2021}
}

@book{santorobook2,
  editor    = {Paola Flocchini and
               Giuseppe Prencipe and
               Nicola Santoro},
  title     = {Distributed Computing by Mobile Entities, Current Research in Moving and Computing},
   publisher = {Springer},
  year      = {2019},
  }

@article{di2017mutual,
  title={Mutual visibility by luminous robots without collisions},
  author={Di Luna, Giuseppe Antonio and Flocchini, Paola and Chaudhuri, S Gan and Poloni, Federico and Santoro, Nicola and Viglietta, Giovanni},
  journal={Information and Computation},
  volume={254},
  pages={392--418},
  year={2017},
  publisher={Elsevier}
}

@inproceedings{sharma2017constant,
  title={Constant-time complete visibility for asynchronous robots with lights},
  author={Sharma, Gokarna and Vaidyanathan, Ramachandran and Trahan, Jerry L},
  booktitle={International Symposium on Stabilization, Safety, and Security of Distributed Systems},
  pages={265--281},
  year={2017},
  organization={Springer}
}

\end{document}